\documentclass[a4paper,11pt]{article}

\usepackage[utf8]{inputenc}
\usepackage[english]{babel}
\usepackage{amsfonts}
\usepackage{amsmath}
\usepackage{amssymb}
\usepackage{amsthm}
\usepackage{yfonts}
\usepackage{indentfirst}
\usepackage{braket}
\usepackage{graphicx}
\usepackage{xcolor}
\usepackage{enumitem}
\usepackage[hidelinks]{hyperref}
\usepackage{indentfirst}
\usepackage{subfiles}
\usepackage{tikz-cd}
\usepackage{scalerel}
\usepackage{mathtools}
\usepackage{authblk}

\setlist[enumerate,1]{label={(\roman*)}}

\numberwithin{equation}{section}
\renewcommand{\phi}{\varphi}
\renewcommand{\theta}{\vartheta}
\renewcommand{\epsilon}{\varepsilon}

\def\dd{\mathrm{d}}
\def\id{\operatorfont{id}}

\def\tr{\operatorname{Tr}}
\def\dim{\operatorname{dim}}
\def\ran{\operatorname{ran}}

\def\supp{\operatorname{supp}}

\def\hilb{\mathcal{H}}

\def\integers{\mathbb{Z}}
\def\egy{1}
\def\geo{\mathrm{geo}}

\newcommand{\weightmtxfn}[1]{\mathbf{w}_{#1}}
\newcommand{\mtxfn}[1]{\mathbf{#1}}
\newcommand{\symfn}[1]{\hat{#1}}
\newcommand{\mtxsymfn}[1]{\hat{\mathbf{#1}}}

\newcommand{\inner}[2]{\Braket{#1|#2}}

\DeclarePairedDelimiter{\norm}{\lVert}{\rVert}

\newcommand{\fock}[1]{\Gamma\!\left(#1\right)}
\newcommand{\car}[1]{\mathrm{CAR}\!\left(#1\right)}

\newcommand{\fockann}[1]{a(#1)}

\def\A{\mathcal{A}}
\def\B{\mathcal{B}}

\def\cD{\mathcal{D}}

\def\F{\mathcal{F}}

\def\M{\mathcal{M}}
\def\T{\mathcal{T}}
\def\P{\mathcal{P}}
\def\S{\mathcal{S}}
\def\X{\mathcal{X}}
\def\bN{\mathbb{N}}
\def\bZ{\mathbb{Z}}
\def\bR{\mathbb{R}}
\def\bC{\mathbb{C}}
\def\bP{\mathbb{P}}
\def\bT{\mathbb{T}}

\def\bz{\left(}
\def\jz{\right)}
\def\hil{\mathcal{H}}
\def\kil{\mathcal{K}}

\def\what{\widehat}

\def\inv{^{-1}}

\def\valt{\cdot}
\def\sa{\mathrm{sa}}

\def\nn{\nonumber}

\def\ep{\varepsilon}
\def\p{_{\ge 0}}
\def\pp{_{>0}}
\def\pne{_{\gneq 0}}
\def\hs{\mathrm{hs}}
\def\tp{\gamma}
\def\op{\mathrm{op}}

\def\n{[n]^*}
\def\divv{\Delta}
\def\qf{\mathrm{qf}}
\def\fdd{^{[1]}}
\def\povm{\mathrm{POVM}}
\def\meas{\mathrm{meas}}

\def\reg{\mathrm{reg}}
\def\dli{\underline{\mathrm{d}}}
\def\dls{\overline{\mathrm{d}}}
\def\dl{\mathrm{d}}
\def\scli{\underline{\mathrm{sc}}}
\def\scs{\overline{\mathrm{sc}}}
\def\sc{\mathrm{sc}}

\newcommand{\ds}{\mbox{ }\mbox{ }}
\newcommand{\ki}{\emph}

\newcommand{\oll}[1]{\overline{#1}}
\newcommand{\ul}[1]{\underline{#1}}

\newcommand{\diad}[2]{|#1\rangle\langle#2|}

\newcommand{\pr}[1]{|#1\rangle\langle #1|}
\newcommand{\canb}[1]{\egy_{\{#1\}}}

\newcommand{\bnorm}[1]{\big\|#1\big\|}
\newcommand{\Bnorm}[1]{\Big\|#1\Big\|}
\newcommand{\opnorm}[1]{\|#1\|_{\infty}}
\newcommand{\W}[1]{W_{#1}}
\newcommand{\fockop}[1]{\Gamma\!\bz #1\jz}
\newcommand{\anticomm}[2]{\left\{#1,#2\right\}}
\newcommand{\fv}{\hat}
\newcommand{\vecc}[1]{\underline{#1}}

\newcommand{\hsd}[2]{\redi_{#1,\hs}}
\newcommand{\hsq}[2]{\req_{#1,\hs}}

\newcommand\err{\epsilon}

\newcommand{\errm}[4]{\epsilon_{\mathrm{mix},#1}^{#2}(#3\|#4)}

\newcommand{\sequence}[1]{\vec{#1}}
\newcommand{\floor}[1]{\left\lfloor #1\right\rfloor}
\newcommand{\ceil}[1]{\left\lceil #1\right\rceil}

\def\d{[d]^*}

\DeclareMathOperator{\Tr}{Tr}

\DeclareMathOperator{\spann}{span}

\DeclareMathOperator*{\medoplus}{\scalerel*{\oplus}{\textstyle\sum}}
\DeclareMathOperator{\ft}{\mathcal{F}}
\DeclareMathOperator{\spec}{spec}

\DeclareMathOperator{\req}{\mathcal{Q}}
\DeclareMathOperator{\red}{D}
\DeclareMathOperator{\redi}{\mathcal{D}}

\DeclareMathOperator{\ootimes}{\otimes\ldots\otimes}

\DeclareMathOperator{\conv}{conv}

\DeclareMathOperator{\essup}{ess \, sup}
\DeclareMathOperator{\dom}{dom}

\theoremstyle{definition}
\newtheorem{definition}{Definition}[section]

\theoremstyle{plain}

\newtheorem{theorem}[definition]{Theorem}
\newtheorem{thm}[definition]{Theorem}
\newtheorem{corollary}[definition]{Corollary}
\newtheorem{cor}[definition]{Corollary}

\newtheorem{lemma}[definition]{Lemma}

\theoremstyle{remark}

\newtheorem{rem}[definition]{Remark}

\title{\textsc{R\'enyi divergences and binary state discrimination error exponents for fermionic quasi-free states}}
\author[1,2]{G\'abor Mar\'oti-Zareczky}
\author[2]{Mil\'an Mosonyi}
\affil[1]{
HUN-REN Wigner Research Centre for Physics, Konkoly-Thege Miklós út 29-33. H-1121 Budapest, Hungary}
\affil[2]{
Department of Analysis and Operations Research, Institute of Mathematics,
Budapest University of Technology and Economics,
M\H uegyetem rkp.~3., H-1111 Budapest, Hungary}

\date{}

\begin{document}

\maketitle

\begin{abstract}
The trade-off relations between the two types of error probabilities in binary i.i.d. quantum state discrimination can be expressed by single-copy formulas in terms of the Petz-type and the sandwiched R\'enyi divergences of the two states representing the two hypotheses. In the non-i.i.d. setting, the error exponents can usually be expressed in terms of regularized Rényi divergences, which do not admit explicit formulas in general. Here, we consider a class of states, translation-invariant and gauge-invariant quasi-free states 
on doubly infinite fermionic chains, and give explicit formulas for a wide range of regularized Rényi divergences between such states, including $(\alpha,z)$, log-Euclidean, geometric, measured, and the recently introduced hockey-stick Rényi divergences. We show that the case where there is a single mode at each lattice site becomes asymptotically classical, with all the different types of regularized Rényi divergences being equal, while in the case of multiple modes per site, non-commutativity persists under regularization, and 
for any fixed $\alpha$, the regularized R\'enyi $(\alpha,z)$-divergences
take different values for different $z$ parameters in general.
Using these results, we give 
explicit expressions for the direct and the strong converse exponents of discriminating two such states when their symbol operators are strictly bounded away from  zero and the identity.
We also generalize a previous construction from [Bunth, Mar\'oti, Mosonyi, Zimbor\'as, Lett.~Math.~Phys.~113:(7), 2023] to the case of multiple modes per lattice site to obtain
a large class of states exhibiting super-exponential decay of the discrimination error probabilities.
\end{abstract}

\tableofcontents

\section{Introduction}

In the problem of binary state discrimination, an experimenter has to read out the value of 
one bit encoded into the state of a quantum system, say, state $\rho$ for bit value $0$ and 
state $\sigma$ for bit value $1$. This is achieved by a measurement with outcomes $0$ and $1$, which, 
in the most general case, the experimenter is free to choose in order to minimize the probability 
of an erroneous decoding. More precisely, if the measurement operators are 
$T_0=T$ and $T_1=I-T$, 
the probability of misidentifying the message $0$ (type I error) is given by 
$\err_0(\rho|T):=\Tr\rho(I-T)$, and 
the probability of misidentifying the message $1$ (type II error) is given by 
$\err_1(\sigma|T):=\Tr\sigma T$. These cannot both be made $0$ unless the two states
are orthogonal to each other, and in general, one error probability can be decreased by changing the measurement only at the expense of increasing the other error probability, showing a trade-off between the values of the two error probabilities. 

The decoding errors can be decreased by sending the same message multiple times, resulting in the code states 
$\rho_n=\rho^{\otimes n}$ (for $0$) and $\sigma_n=\sigma^{\otimes n}$ (for $1$), provided that the encoding 
device has no memory and its operation does not change over time. The two types of error probabilities then 
can be made to go to $0$ with an exponential speed in the number of repetitions $n$ by the right choice of 
measurements, and their trade-off can be quantified on the level of the exponents as
\begin{align}
\dl_r(\sequence{\rho}\|\sequence{\sigma})&:=
\sup\left\{\lim_{n\to+\infty}-\frac{1}{n}\log\err_0(\rho_n|T_n)\,\Big|\,
\lim_{n\to+\infty}-\frac{1}{n}\log\err_1(\sigma_n|T_n)>r
\right\}\label{eq:direct exp1}\\
&=\sup_{\alpha\in(0,1)}\frac{\alpha-1}{\alpha}\left[r-\red_{\alpha,1}(\rho\|\sigma)\right]=:H_r(\rho\|\sigma),
\label{eq:Hoeffding1}
\end{align}
as was shown in \cite{ANSzV,Hayashicq,Nagaoka}.
Here, $\sequence{\rho}:=(\rho_n)_{n\in\bN}$, $\sequence{\sigma}:=(\sigma_n)_{n\in\bN}$, 
the supremum is taken over all 
test sequences $(T_n)_{n\in\bN}$ satisfying the given constraint, and 
$\red_{\alpha,1}(\rho\|\sigma)$ is the Petz-type R\'enyi $\alpha$-divergence of $\rho$ and $\sigma$
\cite{P86}; see Sections \ref{sec:finitedim} and \ref{sec:errexp} for more formal definitions.

In particular, both error probabilities can be made to disappear with an exponential speed as long as the exponent $r$ of the type II error is such that $H_r(\rho\|\sigma)>0$, which is known to be equivalent to 
$r<D(\rho\|\sigma)$, where $D(\rho\|\sigma)$ is the Umegaki relative entropy \cite{Umegaki} of $\rho$ and
$\sigma$. If $r>D(\rho\|\sigma)$ then the type I errors inevitably go to $1$ exponentially fast, known as the strong converse property, and in this case the trade-off is quantified as
\begin{align}
\sc_r(\sequence{\rho}\|\sequence{\sigma})&:=
\inf\left\{\lim_{n\to+\infty}-\frac{1}{n}\log(1-\err_0(\rho_n|T_n))\,\Big|\,
\lim_{n\to+\infty}-\frac{1}{n}\log\err_1(\sigma_n|T_n)>r
\right\}\label{eq:sc exp1}\\
&=\sup_{\alpha>1}\frac{\alpha-1}{\alpha}\left[r-\red_{\alpha,\alpha}(\rho\|\sigma)\right]=:H_r^*(\rho\|\sigma),
\label{eq:Hoeffding2}
\end{align}
as was shown in \cite{MO}. Here, $\red_{\alpha,\alpha}(\rho\|\sigma)$ is the sandwiched R\'enyi 
$\alpha$-divergence \cite{Renyi_new,WWY} of $\rho$ and $\sigma$.

Remarkably, the expressions in \eqref{eq:Hoeffding1} and \eqref{eq:Hoeffding2} only involve a single copy of $\rho$ and $\sigma$, and hence may be explicitly computable, at least numerically.
The situation changes when the encoding device is not assumed to be memoryless anymore.
In fact, the direct exponents $\dl_r(\vec{\rho}\|\vec{\sigma})$ in \eqref{eq:direct exp1}
and the strong converse exponents $\sc_r(\vec{\rho}\|\vec{\sigma})$ in \eqref{eq:sc exp1}
may be defined more generally to quantify the trade-off between the two error probabilities in the asymptotic discrimination of an arbitrary pair of sequences of density operators
$(\rho_n)_{n\in\bN}$, $(\sigma_n)_{n\in\bN}$.
It was shown in \cite{HMO2} and \cite{MO-correlated} that the equalities in \eqref{eq:direct exp1}--\eqref{eq:Hoeffding2} still hold under fairly general conditions, with the single-copy R\'enyi divergences in \eqref{eq:Hoeffding1} and \eqref{eq:Hoeffding2} replaced with their regularized versions
\begin{align*}
\red^{\reg}_{\alpha,\gamma}(\sequence{\rho}\|\sequence{\sigma})
&:=
\lim_{n\to+\infty}\frac{1}{n}\red_{\alpha,\gamma}(\rho_n\|\sigma_n),\ds\ds\ds
\gamma\in\{1,\alpha\}.
\end{align*}
More precisely, what was shown in \cite{HMO2,MO-correlated} is that 
\begin{align}
\dl_r(\sequence{\rho}\|\sequence{\sigma})
&=
\sup_{\alpha\in(0,1)}\frac{\alpha-1}{\alpha}\left[r-\red^{\reg}_{\alpha,1}(\sequence{\rho}\|\sequence{\sigma})\right],\label{eq:direct exp2}
\end{align}
holds whenever
$\red^{\reg}_{\alpha,1}(\sequence{\rho}\|\sequence{\sigma})$ exists for every $\alpha\in(0,1)$ and it is a differentiable function of $\alpha$, and 
\begin{align}
\sc_r(\sequence{\rho}\|\sequence{\sigma})
&=
\sup_{\alpha>1}\frac{\alpha-1}{\alpha}\left[r-\red^{\reg}_{\alpha,\alpha}
(\sequence{\rho}\|\sequence{\sigma})\right],\label{eq:sc exp2}
\end{align}
holds whenever
$\red^{\reg}_{\alpha,\alpha}(\sequence{\rho}\|\sequence{\sigma})$ exists for every $\alpha\in(1,+\infty)$, it is a differentiable function of $\alpha$, and the smallest non-zero eigenvalue of $\sigma_n$ does not go to $0$ faster than exponential in $n$. 

While these are natural and conceptually relevant generalizations of the i.i.d.~results in 
\eqref{eq:direct exp1}--\eqref{eq:Hoeffding2}, their practical relevance is limited unless one can 
express the regularized R\'enyi divergences in an explicitly computable form. Examples of classes of states where this is possible include translation-invariant and gauge-invariant quasi-free states of fermionic lattice systems with one fermionic mode at each physical site, as was demonstrated in 
\cite{MHOF,MO-correlated}. Such a state is specified by a translation-invariant
operator (Toeplitz operator) $Q$ on $\ell^2(\bZ)$, such that $Q$ and $I-Q$ are both positive semi-definite,
or equivalently, via Fourier transformation, by a measurable function $\fv q$ on the one-dimensional torus $\bT$, taking values between $0$ and $1$. (For simplicity, here we only consider a one-dimensional lattice, i.e., a chain of fermions; the case of higher-dimensional lattices 
is very similar.) These are called the symbol operator and the symbol function, respectively, 
of the state $\omega_Q$ on the 
CAR (Canonical Anti-commutation Relation) algebra built on the single-particle Hilbert space $\ell^2(\bZ)$. 
The outcome probabilities of any measurement on a length $n$ portion of the chain are 
determined by the quasi-free state with density operator $\what\omega_{Q_n}$ on the fermionic Fock space built on 
$\ell^2([n])$, which is specified by the symbol operator $Q_n=P_nQP_n$, where $P_n$ is the projection from 
$\ell^2(\bZ)$ onto $\ell^2([n])$.

Given two such states $\omega_Q$ and $\omega_R$ encoding $0$ and $1$, respectively, the aim is to 
read out the value of the bit by making a measurement on a finite (say, length $n$) portion of the chain, the outcome probabilities of which are determined by the density operators 
$\rho_n:=\what\omega_{Q_n}$ and $\sigma_n:=\what\omega_{R_n}$. 
As was shown in \cite{MHOF,MO-correlated}, the regularized R\'enyi divergences are then given by 
\begin{align*}
\red_{\alpha,1}^{\reg}(\omega_Q\|\omega_R)
=
\red_{\alpha,\alpha}^{\reg}(\omega_Q\|\omega_R)
=
\frac{1}{2\pi}\int_0^{2\pi}\frac{1}{\alpha-1}\log\left[\fv q(x)^{\alpha}\fv r(x)^{1-\alpha}+
(1-\fv q(x))^{\alpha}(1-\fv r(x))^{1-\alpha}\right]\,\dd x
\end{align*}
under the assumption that both $\fv q$ and $\fv r$ are strictly bounded away from $0$ and $1$ as
$c\le\fv q(x),\fv r(x)\le(1-c)$, $x\in\bT$, for some $c\in(0,1/2)$. 
(Here we identify $\sequence{\rho}=(\omega_{Q_n})_{n\in\bN}$ with $\omega_Q$, and 
$\sequence{\sigma}=(\omega_{R_n})_{n\in\bN}$ with $\omega_R$.)
Moreover, the differentiability conditions are also satisfied, and hence 
\eqref{eq:direct exp2}--\eqref{eq:sc exp2} hold.
These examples are asymptotically classical in the sense that the symbols $Q$ and $R$ 
of the infinite systems
commute (because they are both mapped into multiplication operators by the Fourier transform), and the regularizations of the 
different types of R\'enyi divergences coincide and are determined by the classical (commuting) objects
$\fv q$ and $\fv r$. One might then suspect this asymptotic commutativity to be the reason why the 
regularized R\'enyi divergences can be given in closed forms, which, however, is not true, as we demonstrate here.

In this paper, we consider a generalization of the above problem, where instead of a single mode per site, 
we allow a fixed finite number of modes at each site, resulting in the single-particle Hilbert space 
$\ell^2(\bZ)\otimes\bC^d$ for some $d\in\bN$, and quasi-free states specified by 
block Toeplitz operator symbols $Q,R$ on $\ell^2(\bZ)\otimes\bC^d$, or equivalently,
matrix-valued symbol functions $\mtxsymfn{q}(x),\mtxsymfn{r}(x)\in\bC^{d\times d}$, $x\in\bT$.
Moreover, we consider the regularization of a large variety of quantum R\'enyi divergences in addition to the 
Petz-type and the sandwiched ones considered previously, and show that 
if the symbols are bounded in the positive semi-definite order as 
$cI_d\le \mtxsymfn{q}(x),\mtxsymfn{r}(x)\le(1-c)I_d$, $x\in\bT$, for some $c\in(0,1/2)$, then 
the regularized R\'enyi divergences exist and can be given in the closed form
\begin{align}
\red_{\alpha,q}^{\reg}(\omega_Q\|\omega_R)
=
\frac{1}{2\pi}\int_0^{2\pi}
\red_{\alpha,q}(\what\omega_{\mtxsymfn{q}(x)}\|\what\omega_{\mtxsymfn{r}(x)})\,\dd x.
\end{align}
Here, $\what\omega_{\mtxsymfn{q}(x)}$ and $\what\omega_{\mtxsymfn{r}(x)}$ are density operators of 
quasi-free states of a fermion system with a $d$-dimensional single-particle Hilbert space at each point 
$x$ of the torus, and 
$\red_{\alpha,q}$ may be any R\'enyi $(\alpha,z)$-divergence \cite{AD} (including the Petz-type and the 
sandwiched R\'enyi divergences), the log-Euclidean R\'enyi divergence \cite{MO-cqconv}, or the 
geometric R\'enyi divergence \cite{HiaiMosonyi2017,Matsumoto_newfdiv,PetzRuskai1998}.
We also evaluate the regularized measured R\'enyi divergences as 
\begin{align}
\red^{\reg}_{\alpha,\meas}(\omega_Q\|\omega_R)=
\begin{cases}
\red^{\reg}_{\alpha,\alpha}(\omega_Q\|\omega_R)
=
\frac{1}{2\pi} \int_0^{2\pi}\red_{\alpha,\alpha}(\what\omega_{\mtxsymfn{q}(x)}\|\what\omega_{\mtxsymfn{r}(x)})\,\dd x
,&\alpha\in[1/2,+\infty),\\
\red^{\reg}_{\alpha,1-\alpha}(\omega_Q\|\omega_R)
=
\frac{1}{2\pi} \int_0^{2\pi}\red_{\alpha,1-\alpha}(\what\omega_{\mtxsymfn{q}(x)}\|\what\omega_{\mtxsymfn{r}(x)})\,\dd x
,&\alpha\in(0,1/2],
\end{cases}
\end{align}
and the regularized version of the recently introduced integral, or hockey-stick R\'enyi divergences
\cite{Frenkel_integral,Hirche_Tomamichel_integral} as 
\begin{align*}
\red^{\reg}_{\alpha,\hs}(\omega_Q\|\omega_R)=
\begin{cases}
\red^{\reg}_{\alpha,1}(\omega_Q\|\omega_R)
=
\frac{1}{2\pi} \int_0^{2\pi}\red_{\alpha,1}(\what\omega_{\mtxsymfn{q}(x)}\|\what\omega_{\mtxsymfn{r}(x)})\,\dd x
,&\alpha\in(0,1),\\
\red^{\reg}_{\alpha,\alpha}(\omega_Q\|\omega_R)
=
\frac{1}{2\pi} \int_0^{2\pi}\red_{\alpha,\alpha}(\what\omega_{\mtxsymfn{q}(x)}\|\what\omega_{\mtxsymfn{r}(x)})\,\dd x
,&\alpha\in(1,+\infty).
\end{cases}
\end{align*}
In particular, we show that the regularized Petz-type and sandwiched R\'enyi divergences
are differentiable in $\alpha$ on $(0,1)$ and on $(1,+\infty)$, respectively, and hence
the direct and the strong converse exponents can be expressed as in 
\eqref{eq:direct exp2} and \eqref{eq:sc exp2}.

In a different direction of generalization, 
one may wonder whether the exponential scale is the only reasonable choice 
on which the asymptotics of the error probabilities can be studied, or if it is possible to obtain faster 
convergence to zero. Such super-exponential error decay was demonstrated in \cite{superexp2023}, again by 
translation-invariant and gauge-invariant fermionic quasifree states on a one-dimensional chain with a single 
mode at each site. 
More precisely, it was shown in \cite{superexp2023} that if there exists a non-degenerate sub-interval $[\mu,\nu]$
of the torus on which the symbol function $\fv q$ is constant $0$, while $\fv r$ is constant $1$, then 
\begin{align}\label{eq:superexp}
\err_0(\what\omega_{Q_n}|T_n)\le e^{-cn\log n}\,,\ds\ds\ds
\err_1(\what\omega_{R_n}|T_n)\le e^{-cn\log n}\,,
\end{align}
for some positive constant $c$ and some test sequence $(T_n)_{n\in\bN}$.
The condition imposed on the symbol functions is very rigid, and does not really allow any modification 
that would be useful in exploring the phenomenon of super-exponential state discrimination any further. 

In this paper we generalize the above result to the case where there are $d$ modes at each site of the chain, and show that \eqref{eq:superexp} still holds with some positive constant $c$ and test sequence $(T_n)_{n\in\bN}$, which we explicitly construct, provided that there 
exists a non-degenerate sub-interval $[\mu,\nu]$
of the torus on which the symbol functions $\fv q$ and $\fv r$ are Lipschitz continuous, and one of the following holds:
\begin{enumerate}
\item\label{item:superexp1-2}
 For every $x\in[\mu, \nu]$, $\mtxsymfn{q}(x)$ and $\mtxsymfn{r}(x)$ are orthogonal, i.e., 
$\mtxsymfn{q}(x)\mtxsymfn{r}(x)=0$, and $\mtxsymfn{r}(x)$ is a non-zero projection;

\item\label{item:superexp2-2} 
for every $x\in[\mu, \nu]$, $I_d-\mtxsymfn{q}(x)$ and $I_d-\mtxsymfn{r}(x)$ are orthogonal, i.e., 
$(I_d-\mtxsymfn{q}(x))(I_d-\mtxsymfn{r}(x))=0$, and $I_d-\mtxsymfn{q}(x)$ is a non-zero projection.
\end{enumerate}
While the bound on the speed of convergence to $0$ that we can prove here is the same as the one in 
\cite{superexp2023}, the above construction obviously offers a lot more flexibility to modify the parameters
and to explore potentially different error asymptotics, which, however, we leave for future work.

The structure of the paper is as follows. In Section \ref{sec:prelim} we collect the necessary preliminaries, especially on block Toeplitz operators and the mathematical description of fermionic systems. 
In Section \ref{sec:Szego}, we extend various Szeg\H o-type limit theorems from \cite{MHOF} to the case of block Toeplitz operators. This will provide the main technical ingredient to evaluate the various regularized 
R\'enyi divergences in Section \ref{sec:Renyi}.
In Section \ref{sec:errexp} we prove \eqref{eq:direct exp2}--\eqref{eq:sc exp2}
for the class of states described above.
Finally, in  Section \ref{sec:superexp} we prove the above result on 
super-exponential error asymptotics.

\section{Preliminaries}
\label{sec:prelim}

\subsection{General}

By $\log$ we will denote the natural logarithm, with its extension to $[0,+\infty]$ as 
$\log 0:=-\infty$, $\log+\infty:=+\infty$.
For a natural number $n\in\bN=\{1,2,\ldots\}$, we will use the notations 
$[n]:=\{1,2,\ldots,n\}$, $[n]^*:=\{0,1,\ldots,n-1\}$.

By a Hilbert space we always mean a complex separable Hilbert space.
We will denote the inner product on a Hilbert space by $\inner{\valt}{\valt}$ and follow the convention that 
it is linear in its second and conjugate linear in its first variable. We will also use the Dirac notation:
for any vectors $x,y$ in a Hilbert space $\hil$, the operator $\diad{y}{x}$ is defined by $\diad{y}{x}z:=\inner{x}{z}y$, $z\in\hil$.

For a linear operator $A$ on a Hilbert space $\hil$, we will use the notations 
$\norm{A}:=\norm{A}_{\infty}:=\sup\{\norm{A\psi}:\,\psi\in\hil,\,\norm{\psi}\le 1\}$ for the operator norm, and 
$\B(\hil):=\{A:\,\hil\to\hil\text{ linear },\norm{A}_{\infty}<+\infty\}$ will denote the set of all bounded linear operators on $\hil$. 
We will use the notation
$\B(\hil)_{\sa}$ for the set of self-adjoint operators on $\hil$.
For an interval $J\subseteq\bR$, 
$\B(\hil)_{J}:=\{A\in\B(\hil)_{\sa}:\,\spec(A)\subseteq J\}$, i.e., 
it is the set of self-adjoint operators on $\hil$ with their spectra in $J$.
We will use the shorthand notations $\B(\hil)_{\ge 0}:=\B(\hil)_{[0,+\infty)}$
for the set of positive semi-definite (PSD) operators on $\hil$, and 
$\B(\hil)_{>0}:=\B(\hil)_{(0,+\infty)}$ for the set of 
positive definite operators, and we will denote by $\B(\hil)\pne$ the set of non-zero PSD operators on $\hil$.
An inequality $A\le B$ between operators $A,B\in\B(\hil)$ is always interpreted in the L\"owner (or PSD) order, meaning $B-A\in\B(\hil)\p$. Elements of the set
\begin{align*}
\B(\hil)_{[0,1]}:=\set{T\in\B(\hil)_{\sa}|0\le T\le I}
\end{align*}
are called \ki{tests} on $\hil$.

The set of (orthogonal) projections on $\hil$ will be denoted by 
\begin{align*}
\bP(\hil):=\B(\hil)_{\{0,1\}}=\set{P\in\B(\hil)_{\sa}|P^2=P}.
\end{align*}
For a positive semi-definite operator $A\in\B(\hil)\p$, we will use the notation
\begin{align}\label{eq:supppr}
A^0:=\lim_{t\searrow 0}A^t
\end{align}
for the projection onto $\supp A:=(\ker A)^{\perp}$. 

The set of \ki{states} (density operators) on a finite-dimensional Hilbert space $\hil$ is
$\S(\hil):=\set{\rho\in\B(\hil)\p|\Tr\rho=1}$.
For any finite set $\X$, the set of positive operator-valued measures (POVMs) on $\hil$ with outcomes in $\X$ is defined as
\begin{align*}
\povm(\hil,\X):=\Set{(M_x)_{x\in\X}\in\B(\hil)\p^{\X}|\sum\nolimits_{x\in\X}M_x=I}.
\end{align*}
The map $T\mapsto (T,I-T)$ gives an identification between 
$\B(\hil)_{[0,1]}$ and $\povm(\hil,\{0,1\})$.
For any POVM $M\in\povm(\hil,\X)$, the corresponding \ki{measurement channel} $\M$ is defined as
\begin{align*}
\M(A):=\sum_{x\in\X}(\Tr M_xA)\pr{x}\in\B(\ell^2(\X)),\ds\ds\ds A\in\B(\hil).
\end{align*}

For a differentiable function $f$ defined on an interval $J\subseteq\bR$, let 
$f\fdd:\,J\times J\to\bR$ be its \ki{first divided difference function}, defined as
\begin{align*}
f\fdd(a,b):=\begin{cases}
\frac{f(a)-f(b)}{a-b},&a\ne b,\\
f'(a),&a=b,
\end{cases}\ds\ds\ds a,b\in J.
\end{align*}
If $f$ is a continuously differentiable function on an open interval $J\subseteq\bR$ then for any finite-dimensional  Hilbert space $\hil$, $A\mapsto f(A)$ is Fr\'echet differentiable on $\B(\hil)_{J}$, and its 
Fr\'echet derivative $(Df)[A]$ at a point $A\in \B(\hil)_{J}$ is given by 
\begin{align}\label{eq:opfunction derivative}
(Df)[A](Y)=\sum_{a,b\in\spec(A)}f\fdd(a,b)P_a^{A}YP_b^A,\ds\ds\ds Y\in\B(\hil)_{\sa},
\end{align}
where $P^A_a$ denotes the spectral projection of $A$ corresponding to an eigenvalue
$a\in\spec(A)$.
See, e.g., \cite[Theorem V.3.3]{Bhatia} or \cite[Theorem 2.3.1]{Hiai_book}.
It is easy to see from this that if $f$ is as above, and $(a,b)\ni t\mapsto A(t)\in\B(\hil)_{J}$
is continuously differentiable, then so is $\Tr f(A(t))$ as well, and 
\begin{align}\label{eq:Tr derivative}
\frac{d}{dt}\Tr f(A(t))=\Tr \left[f'(A(t))\frac{d}{dt}A(t)\right],\ds\ds\ds t\in(a,b).
\end{align}

\subsection{Block Toeplitz operators}

For a measure space $(\X,\F,\mu)$ and a finite-dimensional Hilbert space $\hil$, 
let 
\begin{align*}
L^2(\X,\hil):=\left\{f\in\hil^{\X}\text{ measurable, }\norm{f}_2^2:=\int_{\X}\norm{f(t)}^2\,\dd\mu(t)<+\infty
\right\},
\end{align*} 
and 
\begin{align}
L^{\infty}(\X,\B(\hil))&:=\big\{A\in\B(\hil)^{\X}\text{ measurable, }\nn\\
&\ds\ds\ds\ds\ds\norm{A}_{\infty}:=
\inf\{C>0:\,\mu(\{t\in\X:\,\opnorm{A(t)}>C\})=0\}<+\infty\big\}\label{eq:bounded opfunctions}\\
&\subseteq\B\bz L^2(\X,\hil)\jz,\nn
\end{align}
where $A\in L^{\infty}(\X,\B(\hil))$ acts on $f\in L^{2}(\X,\hil)$ as 
$(Af)(x):=A(x)f(x)$, $x\in\X$.
It is easy to see that the operator norm of such an operator coincides with its norm defined in 
\eqref{eq:bounded opfunctions}, justifying the same notation for the two.
In particular, when 
$\hil=\bC^d:=\bC^{[d]^*}$ for some $d\in\bN$,  and $\X=\bZ$
or $\X=[n]^*$ for some $n\in\bN$, $\F$ is its full power set, and $\mu$ is the counting measure, we will use the notations
\begin{align*}
\ell^2_d(\X)&:=L^2(\X,\bC^d)=\set{f:\X\to\bC^d|\|f\|_2<\infty},\qquad 
\|f\|_2^2=\sum_{x\in\X}\|f_x\|^2_{\bC^d},\\
\ell^{\infty}_{d\times d}(\X)&:=L^{\infty}(\X,\B(\bC^d)),
\end{align*}
where 
$\|\cdot\|_{\bC^d}$ denotes the usual norm of $\bC^d$. 
We will often use the following natural identifications in the above case: 
\begin{equation*}
    \ell^2_d(\X)\equiv\bigoplus_{k=0}^{d-1}\ell^2(\X)\equiv\ell^2(\X)\otimes\bC^d,
\end{equation*}
which in turn gives that 
any bounded operator $A\in\B(\ell^2_d(\X))$ can be decomposed as
\begin{align*}
    A=\left[ A_{k,l} \right]_{k,l=0}^{d-1}\ds\equiv \ds
    \sum_{k,l=0}^{d-1} A_{k,l}\otimes \diad{k}{l},
\end{align*}
where $A_{k,l}\in\B(\ell^2(\X))$, $k,l\in[d]^*$, and we use the standard shorthand notation
\begin{align*}
\ket{k}:=\ket{\egy_{\{k\}}},\ds\ds\ds k\in\d.
\end{align*}

The translation operator $T$ on $\ell^2(\bZ)$ is given by $T\canb{k}=\canb{k+1}$, $k\in\bZ$, and its extension to $\ell^2_d(\bZ)$ is
\begin{align*}
\T=\bigoplus_{k=0}^{d-1}T\ds\equiv \ds T\otimes I.
\end{align*}
An operator $A\in\B(\ell_d^2(\bZ))$ is said to be \ki{translation-invariant}, 
or a \ki{block Toeplitz operator},
if  
\begin{align}\label{eqn:shift-inv}
\T A\T\inv=A,\ds\ds\ds\text{or equivalently,}\ds\ds\ds
TA_{k,l}T\inv=A_{k,l},\ds k,l\in[d]^*,
\end{align}
i.e., if every one of its blocks $A_{k,l}$, $k,l\in[d]^*$, is translation-invariant
(also called a \ki{Toeplitz operator}).

Let $\bT=[0,2\pi)$ denote the one-dimensional torus equipped with its canonical rotation
(equivalently, modulo $2\pi$ translation) and the Lebesgue measure. 
The Fourier transform is given by  
\begin{align*}
    F:\,\ell^2(\bZ)\to L^2(\bT),\ds\ds
    F\canb{k}:=\chi_k:=\frac{1}{\sqrt{2\pi}}e^{ik(\valt)},\ds\ds\ds k\in\bZ.
\end{align*}
Its canonical extension from $\ell^2_d(\bZ)$ to 
\begin{align*}
L^2_d(\bT):=L^2(\bT,\bC^d)\equiv\bigoplus_{k=0}^{d-1}L^2(\bT)\equiv L^2(\bT)\otimes\bC^d
\end{align*}
is given by 
\begin{align*}
    \ft=\bigoplus_{k=0}^{d-1}F\ds\equiv \ds F\otimes I.
\end{align*}

The matrix elements of a translation-invariant operator $A\in\B(\ell^2(\bZ))$
in the canonical orthonormal basis $\{\canb{k}\}_{k\in\bZ}$ are given by 
\begin{align*}
    \inner{\canb{k}}{A\canb{l}}=\inner{\canb{k}}{T^lAT^{-l}\canb{l}}
    =\inner{\canb{k-l}}{A\canb{0}}=a(k-l),
\end{align*}
where $a(k):=\inner{\canb{k}}{A\canb{0}}$, $k\in\bZ$. Since $a\in\ell^2(\bZ)$, the function
\begin{align*}
    \hat a:=\sqrt{2\pi}\sum_{k\in\bZ}a(k)\chi_k=\sum_{k\in\bZ}a(k)e^{ik(\valt)}
\end{align*}
is well defined as an element of $L^2(\bT)$. Let $M_{\hat a}:\,f\mapsto \hat a f$ denote the corresponding multiplication operator on its natural domain $\dom(M_{\hat a})$ in $L^2(\bT)$. Then
\begin{align*}
    \inner{\canb{k}}{(F^{-1} M_{\hat a}F)\canb{l}}&=
    \inner{\chi_k}{\hat a\chi_l}=\int_0^{2\pi}\frac{1}{2\pi}\hat a(x)e^{i(l-k)x}\,dx\\
    &=\frac{1}{\sqrt{2\pi}}\inner{\chi_{k-l}}{\hat a}=a(k-l)=\inner{\canb{k}}{A\canb{l}}.
\end{align*}
Since this holds for every $k,l\in\bZ$, we get that for all trigonometric polynomial $p$, $p\in\dom(M_{\hat a})$, and $\norm{M_{\hat a}p}=\norm{AF^{-1}p}\le\norm{A}\norm{p}$, whence $M_{\hat a}$ is bounded, and 
\begin{align}\label{eq:tiop Fourier}
A=F^{-1} M_{\hat a}F.
\end{align}
Thus, every translation-invariant operator on $\ell^2(\bZ)$ is mapped into a multiplication operator by the Fourier transform. This implies that any translation-invariant operator is normal, and any two 
translation-invariant operators commute with each other. 
In particular, if $A^{(j)}\in\B(\ell^2_d(\bZ))$, $j=1,2$, are translation-invariant, then 
any of their blocks commute, i.e., $A^{(j)}_{k,l}A^{(j')}_{k',l'}=A^{(j')}_{k',l'}A^{(j)}_{k,l}$, 
$j,j'\in\{1,2\}$, $k,k',l, l'\in[d]^*$. 

Now, if $A \in \B(\ell^2_d(\bZ))$ is translation-invariant
then by (\ref{eqn:shift-inv}) and \eqref{eq:tiop Fourier},
\begin{align*}
    A &= \sum_{k,l=0}^{d-1} A_{k,l}\otimes\ket{k}\bra{l} 
      = \sum_{k,l=0}^{d-1} (F^{-1}M_{\symfn{a}_{kl}}F)\otimes\ket{k}\bra{l} 
      = \ft^{-1} \underbrace{\left(\sum_{k,l=0}^{d-1} M_{\symfn{a}_{kl}}\otimes\ket{k}\bra{l}\right)}_{
      =:M_{\mtxsymfn{a}}}\ft
      =\ft^{-1}M_{\mtxsymfn{a}}\ft,
\end{align*}
where $\symfn{a}_{kl}\in L^\infty(\bT)$ for all $k,l\in\d$, and we introduce 
the notation
\begin{equation*}
    \mtxsymfn{a} := \sum_{k,l=0}^{d-1}\symfn{a}_{kl} \otimes \ket{k}\bra{l}
    \in L^{\infty}(\bT,\B(\bC^d))=:L^{\infty}_{d\times d}(\bT).
\end{equation*}

\subsection{Fermionic systems}

For vectors $\phi_1,\ldots,\phi_k$ in a complex Hilbert space $\hil$, let 
\begin{align}\label{eq:as product}
\phi_1\wedge\ldots\wedge\phi_k:=\frac{1}{\sqrt{k!}}\sum_{\sigma\in \mathfrak{S}_k}\varepsilon(\sigma)\phi_{\sigma(1)}\ootimes\phi_{\sigma(k)}
\end{align}
denote their anti-symmetrized tensor product, where
$\mathfrak{S}_k$ stands for the set of permutations of $k$ elements and $\varepsilon(\sigma)$ 
for the sign of the permutation $\sigma$. For any $k\in\bN$,
the $k$-th anti-symmetric tensor power $\wedge^k\hil=\hil^{\wedge k}$ of $\hil$ is the 
closure of the subspace of 
$\hil^{\otimes k}$ spanned by all vectors of the form \eqref{eq:as product}, and we define
$\hil^{\wedge 0}:=\bC$.

The Hilbert space of a fermionic system with single-particle Hilbert space $\hil$ is
the \ki{anti-symmetric Fock space} (or \ki{fermionic Fock space})
\begin{align*}
\fock{\hil}:=\medoplus_{k=0}^{\dim\hil}\hil^{\wedge k}.
\end{align*}
For any operator $A\in\B(\hil)$ and any $k\in\bN$, $A^{\otimes k}$ leaves the subspace $\hil^{\wedge k}$
invariant, and we define
\begin{align*}
A^{\wedge k}:=A^{\otimes k}\vert_{\hil^{\wedge^k}}
\end{align*}
as an operator on $\hil^{\wedge k}$. 
If $A\in\B(\hil)$ is a contraction or it is compact then 
\begin{align*}
\fockop{A}:=\medoplus_{k=0}^{\dim\hil} A^{\wedge k}
\end{align*}
defines a bounded operator on $\fock{\hil}$, where $A^{\wedge 0}:=1\in\B(\bC)$.
It is easy to see that if $\hil$ is finite dimensional then 
\begin{align}\label{eq:Fockop trace}
\Tr\fockop{A}
=
\det (I+A).
\end{align}
We will also often use the easily verifiable fact that 
if $\hil$ is finite dimensional then
for any positive definite 
$A,B\in\B(\hil)\pp$ and $x,y\in\bR$, 
\begin{align*}
\fockop{A}^x\fockop{B}^y=\fockop{A^xB^y}.
\end{align*}

For each $\phi\in\hilb$, the corresponding \emph{creation operator} 
$c(\phi)$ is the unique bounded linear extension of the map
\begin{align*}
\phi_1\wedge\ldots\wedge\phi_k\mapsto\phi\wedge\phi_1\wedge\ldots\wedge\phi_k,
\ds\ds\ds\phi_1,\ldots,\phi_k\in\hil,
\end{align*}
and the corresponding \emph{annihilation operator} is its adjoint, 
$a(\phi):=c(\phi)^*$. These operators satisfy the \emph{canonical anti-commutation relations (CARs)},
\begin{align} \label{eqn:cars}
\anticomm{a(\phi)}{a(\psi)}= 0, \ds\ds\ds 
\anticomm{a(\phi)}{a^*(\psi)}= \inner{\phi}{\psi}I,\ds\ds \phi,\psi\in\hil.
\end{align}
The $C^*$-subalgebra of $\B(\fock{\hil})$ generated by $\{a(\phi):\,\phi\in\hil\}$ 
is called the \emph{algebra of the canonical anti-commutation relations} (or \emph{CAR-algebra}) corresponding to the single-particle Hilbert space $\hil$, and is denoted by $\car{\hil}$.

A \ki{state} on $\car{\hil}$ is a positive linear functional that takes the value $1$ on $I$.    
For any positive semi-definite operator $Q\in\B(\hil)$ with $Q\le I$ there exists a unique state $\omega_Q$ on $\car{\hil}$ 
(called the \ki{gauge-invariant quasi-free state with symbol $Q$})
with the property
\begin{align} \label{eqn:quasifree_functional}
\omega_Q\left(a(\phi_1)^*\ldots a(\phi_n)^*a(\psi_m)\ldots a(\psi_1)\right) 
= \delta_{mn}\det\left\{\Braket{\psi_i|Q\phi_j}\right\}_{i,j=1}^n.
\end{align}
It is easy to verify that when $d:=\dim\hil$ is finite, the density operator 
$\what\omega_Q$ of $\omega_Q$ can be explicitly given as
\begin{align}\label{quasifree density}
\what\omega_Q=\prod_{j=1}^d\bz q_ja(e_j)^*a(e_j)+(1-q_j)a(e_j)a(e_j)^*\jz
\end{align}
where $Q=\sum_{j=1}^d q_j\pr{e_j}$ is any eigen-decomposition of $Q$.
Note that for all $1\le i_1<\ldots<i_k\le d$, $e_{i_1}\wedge\ldots\wedge e_{i_k}$
is an eigenvector of $\what\omega_Q$ with eigenvalue
$\bz\prod_{j\in\{i_1,\ldots,i_k\}}q_j\jz\cdot\bz \prod_{j\in[d]\setminus\{i_1,\ldots,i_k\}}(1-q_j)\jz$. 
This implies immediately that if $1$ is not an eigenvalue of $Q$ then 
$\what\omega_Q$ can be written as 
\begin{align}\label{eq:quasifree density}
\what\omega_Q=\det(I-Q)\fockop{\W{Q}},\ds\ds\ds\ds\ds
\W{Q}:=\frac{Q}{I-Q}.
\end{align} 

Since in this paper we only consider quasi-free states that are gauge-invariant, 
in the following we will drop
``gauge-invariant'' from the terminology, i.e., by a 
quasi-free state we will always mean a 
gauge-invariant quasi-free state.

Quasi-free states emerge as equilibrium states of non-interacting fermionic systems. For instance, if the single-particle Hamiltonian $H$ of a system of non-interacting fermions 
is such that $e^{-\beta H}$ is trace-class 
then 
the Gibbs state of the system at inverse temperature $\beta$ is the quasi-free state with symbol
$Q=\frac{e^{-\beta H}}{I+e^{-\beta H}}$ (see, e.g., \cite[Proposition 5.2.23]{BR2}).

Consider now a fermionic chain with $d$ modes at each site, the
single-particle Hilbert space of which is $\hil=\ell^2_d(\integers)$.
The translation operator $\T$ on $\ell^2_d(\bZ)$ defines the 
\ki{translation automorphism} $\tau$ on $\car{\ell^2_d(\bZ)}$ via
$\tau(\fockann{\phi}):=\fockann{\T\phi}$, $\phi\in\ell^2_d(\bZ)$.
A quasi-free state $\omega_Q$ on $\car{\ell^2_d(\bZ)}$ is called \emph{translation-invariant} if 
$\omega_Q\circ\tau=\omega_Q$, which is easily seen to be equivalent to 
$\T Q\T\inv=Q$, i.e., the translation-invariance of the symbol $Q\in\B(\ell^2_d(\bZ))$. For instance, in the above example a translation-invariant single-particle Hamiltonian $H$ yields a translation-invariant quasi-free state as the equilibrium state of the system. 

A measurement on a subsystem corresponding to modes at the sites $\n:=\{0,\ldots,n-1\}$ has
measurement operators in the $C^*$-subalgebra 
$\A_n\subseteq\car{\ell^2_d(\bZ)}$ generated by $\{a(\phi):\,\phi\in\hil_n\}$,
\begin{align*}
\hil_n:=\ell^2_d(\n)\equiv\spann\{\egy_{\{k\}}\otimes\egy_{\{j\}}:\,k\in\n,\,j\in\d\}\subseteq 
\ell^2(\bZ)\otimes\bC^d\equiv
\ell^2_d(\bZ).
\end{align*}
This subalgebra is naturally isomorphic to $\car{\ell^2_d(\n)}$.
It is easy to see that if the 
state of the infinite chain is given by a quasi-free state with symbol $Q$ then 
the statistics of any such local measurement is given by the quasi-free state 
$\omega_{Q_n}$
with symbol 
\begin{align}\label{eq:cutoff def}
Q_n:=(P_n\otimes I_d)^*Q(P_n\otimes I_d), \ds\ds\ds\text{where}\ds\ds\ds
P_n:=\sum_{k=0}^{n-1}\pr{\egy_{\{k\}}}.
\end{align}

\section{Szegő-type theorems for block Toeplitz operators}
\label{sec:Szego}

In this section, we will consider generalizations of the 
Szeg\H o-type result given in Lemma \ref{lem:szego-products}. For a proof of the latter, see \cite{MHOF}.

\begin{lemma}\label{lem:szego-products}
    Let $\symfn{a}^{(1)},\dots,\symfn{a}^{(r)}\in L^\infty(\bT)$ with the corresponding translation-invariant operators $A^{(k)}=F\inv M_{\mtxsymfn{a}^{(k)}}F\in\B(\ell^2(\bZ))$. Then
\begin{equation*}
        \lim_{n\to\infty}\frac{1}{n}\tr A^{(1)}_n\cdots A^{(r)}_n = \frac{1}{2\pi}\int_0^{2\pi}\symfn{a}^{(1)}(x)\cdots\symfn{a}^{(r)}(x) \,\dd x,
    \end{equation*}
    where $A_n^{(m)}:= P_n A^{(m)}P_n$, $m\in[r]$, with $P_n := \sum_{k=0}^{n-1} \Ket{\canb{k}}\Bra{\canb{k}}$.
\end{lemma}

\begin{lemma}\label{lemma:szego-block-prod}
Let $\mtxsymfn{a}^{(1)},\dots,\mtxsymfn{a}^{(r)}\in L^\infty_{d\times d}(\bT)$ with the corresponding translation-invariant operators 
$A^{(k)}=\ft\inv M_{\mtxsymfn{a}^{(k)}}\ft\in\B(\ell^2_d(\bZ))$, $k\in[r]$. Then
    \begin{equation}\label{eqn:szego-block-prod}
        \lim_{n\to\infty}\frac{1}{n}\tr A^{(1)}_n \cdots A^{(r)}_n = \frac{1}{2\pi}\int_0^{2\pi}\tr\mtxsymfn{a}^{(1)}(x)\cdots\mtxsymfn{a}^{(r)}(x) \,\dd x,
    \end{equation}
    where $A_n^{(m)}:= \P_n A^{(m)}\P_n$, $m\in[r]$, with $\P_n := \bigoplus_{k=0}^{d-1}P_n \equiv P_n\otimes I$ and $P_n := \sum_{k=0}^{n-1} \Ket{\canb{k}}\Bra{\canb{k}}$.
\end{lemma}

\begin{proof}
    For all $m\in[r]$, we have 
    \begin{align*}
        A_n^{(m)} &= \P_n A^{(m)} \P_n 
                  = \sum_{k,l=0}^{d-1}P_nA_{k,l}^{(m)}P_n\otimes\ket{k}\bra{l} 
                  = \sum_{k,l=0}^{d-1} \left(A^{(m)}_{k,l}\right)_n \otimes\ket{k}\bra{l},
    \end{align*}
    whence
    \begin{align*}
        \tr A^{(1)}_n A^{(2)}_n\cdots A^{(r)}_n = \tr \sum_{k_1,\dots,k_r=0}^{d-1}\left(A^{(1)}_{k_r,k_1}\right)_n\left(A^{(2)}_{k_1,k_2}\right)_n\cdots\left(A^{(r)}_{k_{r-1},k_r}\right)_n.
    \end{align*}
    Hence, by the linearity of the trace and the limit, we can rewrite the left-hand side of (\ref{eqn:szego-block-prod}) as
    \begin{align*}
 &\sum_{k_1,\dots,k_r=0}^{d-1} \lim_{n\to\infty}\frac{1}{n}\tr\left(A^{(1)}_{k_r,k_1}\right)_n\left(A^{(2)}_{k_1,k_2}\right)_n\cdots\left(A^{(r)}_{k_{r-1},k_r}\right)_n\\
 &\ds=
        \sum_{k_1,\dots,k_r=0}^{d-1}\frac{1}{2\pi}\int_0^{2\pi}\symfn{a}_{k_rk_1}^{(1)}(x)\symfn{a}_{k_1k_2}^{(2)}(x)\cdots\symfn{a}_{k_{r-1}k_r}^{(r)}(x) \,\dd x\\
 &\ds=
  \frac{1}{2\pi}\int_0^{2\pi}
  \sum_{k_1,\dots,k_r=0}^{d-1}\symfn{a}_{k_rk_1}^{(1)}(x)\symfn{a}_{k_1k_2}^{(2)}(x)\cdots\symfn{a}_{k_{r-1}k_r}^{(r)}(x) \,\dd x\\
 &\ds=
  \frac{1}{2\pi}\int_0^{2\pi}
\Tr
\mtxsymfn{a}^{(1)}(x)\mtxsymfn{a}^{(2)}(x)\cdots\mtxsymfn{a}^{(r)}(x) 
\,\dd x,
\end{align*}
where the first equality follows from Lemma \ref{lem:szego-products}, 
and the rest are obvious.
\end{proof}

\begin{theorem}\label{thm:szego-block-poly}
Let $\mtxsymfn{a}^{(1)},\dots,\mtxsymfn{a}^{(r)}\in L^\infty_{d\times d}(\bT)$ with the corresponding translation-in\-va\-ri\-ant operators 
$A^{(k)}=\ft\inv M_{\mtxsymfn{a}^{(k)}}\ft\in\B(\ell^2_d(\bZ))$, $k\in[r]$. Then
\begin{align}\label{eqn:szego-block-poly}
\lim_{n\to\infty}\frac{1}{n}\tr f^{(1)}(A^{(1)}_n) \cdots f^{(r)}(A^{(r)}_n) 
 = \frac{1}{2\pi}\int_0^{2\pi}\tr f^{(1)}(\mtxsymfn{a}^{(1)}(x))\cdots f^{(r)}(\mtxsymfn{a}^{(r)}(x)) \,\dd x,
\end{align}
for any choice of polynomials $f^{(1)},\ldots,f^{(r)}$. If, moreover, each $\mtxsymfn{a}^{(k)}$ is self-adjoint almost everywhere, then (\ref{eqn:szego-block-poly}) holds also when each $f^{(k)}$ is a continuous function on 
$\cD_k:=\operatorname{conv}\bz\spec\left(A^{(k)}\right)\jz$.
\end{theorem}

\begin{proof}
The statement for polynomials follows immediately from Lemma \ref{lemma:szego-block-prod}. 
    
Now, if the $\mtxsymfn{a}^{(k)}$ are self-adjoint almost everywhere, then the $A^{(k)}_n$ are also self-adjoint for all $n\in\bN$. Moreover, the spectrum of $A^{(k)}_n$ is easily seen to be contained in the convex hull of the spectrum of $A^{(k)}$. Therefore, $f^{(k)}(A^{(k)}_n)$ is well defined for all $n$. 
By a simple application of the Stone-Weierstrass approximation theorem, one obtains that
for every $\varepsilon>0$ there exist polynomials $f^{(1)}_{\varepsilon},\dots,f^{(r)}_{\varepsilon}$ such that
    \begin{equation*}
        \bnorm{f^{(k)}-f^{(k)}_{\varepsilon}}_{\infty} < \varepsilon 
        \qquad\textrm{and}\qquad 
        \bnorm{f^{(k)}_{\varepsilon}}_{\infty} \leq \bnorm{f^{(k)}}_{\infty},\ds\ds\ds
        k\in[r],
    \end{equation*}
where for $\#=\set{}$ and $\#=\varepsilon$,
\begin{align*}
\bnorm{f^{(k)}_{\#}}_{\infty}:=\max_{t\in\cD_k}\big|f^{(k)}_{\#}(t)\big|\,.
\end{align*}
    Moreover, for all $k\in[r]$ and $\#=\set{}$ and $\#=\varepsilon$, we have
    \begin{align*}
        \bnorm{f^{(k)}_{\#}(A^{(k)}_n)} 
        \leq 
        \bnorm{f^{(k)}_{\#}}_{\infty},
        \ds\ds\ds
        \bnorm{f^{(k)}_{\#}(\mtxsymfn{a}^{(k)})}_{\infty} 
        \leq 
        \bnorm{f^{(k)}_{\#}}_{\infty}\,.
    \end{align*}
   
    Let $I_{nd}$ denote the identity operator on $\ran \P_n$. Then, for any bounded operators $X,Y$ on $\ran \P_n$ and any continuous functions $f,g$ on $\operatorname{conv}\bz\sigma\left(A^{(k)}\right)\jz$,
   we have 
    \begin{align}
        \left|\tr X f(A^{(k)}_n)Y-\tr X g(A^{(k)}_n)Y\right| 
        &\leq 
        \Bnorm{X \left[f(A^{(k)}_n)-g(A^{(k)}_n)\right]Y}_1 \nonumber \\
        &\leq 
        \norm{X}\norm{Y} \bnorm{f(A^{(k)}_n)-g(A^{(k)}_n)}_1 \nonumber \\
        &\leq 
        \norm{X}\norm{Y}\norm{I_{nd}}_1 \bnorm{f(A^{(k)}_n)-g(A^{(k)}_n)}_{\infty} \nonumber \\
        &\leq nd\norm{X}\norm{Y}\|f-g\|_{\infty}, \label{eqn:holder-ineq-1}
    \end{align}
    where in the 
    third inequality we used the Hölder inequality. Similarly, for any bounded operators $Z, W$ on $\bC^d$ and any continuous functions $f,g$ defined on $\operatorname{conv}\bz\spec\left(M_{\mtxsymfn{a}^{(k)}}\right)\jz$, we have
\begin{align}
        \left|\tr Z f(\mtxsymfn{a}^{(k)}(x))W-\tr Z g(\mtxsymfn{a}^{(k)}(x))W\right| 
        &\leq 
        \left\|Z \left[f(\mtxsymfn{a}^{(k)}(x))-g(\mtxsymfn{a}^{(k)}(x))\right]W\right\|_1 \nonumber \\
        &\leq 
        \|Z\| \|W\| \bnorm{f(\mtxsymfn{a}^{(k)}(x))-g(\mtxsymfn{a}^{(k)}(x))}_1 \nonumber \\
        &\leq 
        \|Z\| \|W\| \|I_d\|_1 \bnorm{f(\mtxsymfn{a}^{(k)}(x))-g(\mtxsymfn{a}^{(k)}(x))}_{\infty} \nonumber \\
        &\leq d\|Z\| \|W\|\|f-g\|_{\infty}, \label{eqn:holder-ineq-2}
    \end{align}
for almost every $x\in\bT$.
    
    With repeated use of (\ref{eqn:holder-ineq-1}) and (\ref{eqn:holder-ineq-2}), one can see that
    \begin{align}\label{eqn:szego-block-poly proof1}
        \left|\frac{1}{n}\tr f^{(1)}(A^{(1)}_n) \cdots f^{(r)}(A^{(r)}_n)
        -\frac{1}{n}\tr f^{(1)}_{\varepsilon}(A^{(1)}_n)\cdots f^{(r)}_{\varepsilon}(A^{(r)}_n)\right| \leq \varepsilon d r \max_{1\leq k\leq r} \bnorm{f^{(k)}}_{\infty}^{r-1},
    \end{align}
    and 
    \begin{align}\label{eqn:szego-block-poly proof2}
        \left|\tr f^{(1)}(\mtxsymfn{a}^{(1)}(x))\cdots f^{(r)}(\mtxsymfn{a}^{(r)}(x))
-\tr f^{(1)}_{\varepsilon}(\mtxsymfn{a}^{(1)}(x))\cdots f^{(r)}_{\varepsilon}(\mtxsymfn{a}^{(r)}(x))\right| \leq \varepsilon d r \max_{1\leq k\leq r} \bnorm{f^{(k)}}_{\infty}^{r-1},
    \end{align}
    for almost every $x\in\bT$.
 
As we have established above, \eqref{eqn:szego-block-poly} holds when all the $f^{(k)}$ are polynomials, whence
for every $\ep>0$,
\begin{align}\label{eqn:szego-block-poly proof3}
\lim_{n\to+\infty}
\frac{1}{n}\tr f^{(1)}_{\varepsilon}(A^{(1)}_n) \cdots f^{(r)}_{\varepsilon}(A^{(r)}_n)
=
\frac{1}{2\pi}\int_0^{2\pi}\tr f^{(1)}_{\varepsilon}(\mtxsymfn{a}^{(1)}(x)) 
\cdots f^{(r)}_{\varepsilon}(\mtxsymfn{a}^{(r)}(x))\,\dd x.
\end{align}    
Combining \eqref{eqn:szego-block-poly proof1}--\eqref{eqn:szego-block-poly proof3} then yields
\eqref{eqn:szego-block-poly} by a straightforward argument.
\end{proof}

\begin{corollary}\label{cor:for-renyi-alpha-z}
In the setting of Theorem \ref{thm:szego-block-poly}, 
let $f^{(k)}$ be continuous functions on $\operatorname{conv}\bz\spec(A^{(k)})\jz$ for $k=1,\dots,r$, and define
    \begin{align}\label{eq:B def}
    	B_n := \prod_{k=1}^r f^{(k)}(A^{(k)}_n),\ds\ds n\in\bN,
    	\ds\ds\ds\ds\ds
    	\mtxfn{b}(x) := \prod_{k=1}^r f^{(k)}(\mtxsymfn{a}^{(k)}(x)),\ds\ds x\in\bT.
    \end{align}
Then, for any continuous function $g: [0, M] \to \bR$, where 
$M:= \prod_{k=1}^r \max_{s \in \conv(\spec(A^{(k)}))} |f^{(k)}(s)|^2$, we have
    \begin{equation}\label{eqn:for-renyi-alpha-z}
        \lim_{n\to\infty} \frac{1}{n} \tr g\left(B_n B_n^*\right) 
        = 
        \frac{1}{2\pi} \int_0^{2\pi} \tr g\left(\mtxfn{b}(x)\mtxfn{b}(x)^*\right) \,\dd x.
    \end{equation}
\end{corollary}

\begin{proof}
    Since $g$ is continuous on the compact interval $[0,M]$, the Stone-Weierstrass theorem guarantees that for every $\varepsilon > 0$, there exists a polynomial $g_\varepsilon$ such that
    \begin{equation*}
        \|g - g_\varepsilon\|_\infty := \max_{y \in [0,M]} |g(y) - g_\varepsilon(y)| < \varepsilon.
    \end{equation*}
    
    For a fixed $\varepsilon > 0$, $g_\varepsilon(B_n B_n^*)$ evaluates to a finite sum of products of factors of the form $f^{(k)}(A^{(k)}_n)$. Therefore, Theorem \ref{thm:szego-block-poly} ensures 
    that there exists an $N_\varepsilon \in \bN$ such that for all $n \geq N_\varepsilon$, 
    \begin{equation}\label{for-renyi-alpha-z proof1}
        \left| \frac{1}{n} \tr g_\varepsilon\left(B_n B_n^*\right) - \frac{1}{2\pi} \int_0^{2\pi} \tr g_\varepsilon\left(\mtxfn{b}(x)\mtxfn{b}(x)^*\right) \,\dd x \right| < \varepsilon.
    \end{equation}
    
Using the H\"older inequality, we get
    \begin{align}
            \left| \frac{1}{n} \tr g\left(B_n B_n^*\right) - \frac{1}{n} \tr g_\varepsilon\left(B_n B_n^*\right) \right| 
        &\le \frac{1}{n} \left\| g\left(B_n B_n^*\right) - g_\varepsilon\left(B_n B_n^*\right) \right\|_1 \nn\\
        &\le \frac{1}{n} \|I_{nd}\|_1 \left\|g\left(B_n B_n^*\right) - g_\varepsilon\left(B_n B_n^*\right)\right\|_{\infty} \nn\\
        &\le \frac{nd}{n} \|g - g_\varepsilon\|_\infty < d\varepsilon.
\label{for-renyi-alpha-z proof2}
\end{align}
Similarly, 
\begin{align}
        &\left| \frac{1}{2\pi} \int_0^{2\pi} \tr g\left(\mtxfn{b}(x)\mtxfn{b}(x)^*\right) \,\dd x - \frac{1}{2\pi} \int_0^{2\pi} \tr g_\varepsilon\left(\mtxfn{b}(x)\mtxfn{b}(x)^*\right) \,\dd x \right| \nn\\
        &\le \frac{1}{2\pi} \int_0^{2\pi} \left\| g\left(\mtxfn{b}(x)\mtxfn{b}(x)^*\right) - g_\varepsilon\left(\mtxfn{b}(x)\mtxfn{b}(x)^*\right) \right\|_1 \,\dd x \nn\\
        &\le \frac{1}{2\pi} \int_0^{2\pi} d \left\|g\left(\mtxfn{b}(x)\mtxfn{b}(x)^*\right) - g_\varepsilon\left(\mtxfn{b}(x)\mtxfn{b}(x)^*\right)\right\|_{\infty} \,\dd x < d\varepsilon.
        \label{for-renyi-alpha-z proof3}
    \end{align}
    
Combining \eqref{for-renyi-alpha-z proof1}--\eqref{for-renyi-alpha-z proof3}, we get that for every $n\ge N_{\ep}$, 
    \begin{align*}
        &\left| \frac{1}{n} \tr g\left(B_n B_n^*\right) - \frac{1}{2\pi} \int_0^{2\pi} \tr g\left(\mtxfn{b}(x)\mtxfn{b}(x)^*\right) \,\dd x \right| \\
        &\le \left| \frac{1}{n} \tr g\left(B_n B_n^*\right) - \frac{1}{n} \tr g_\varepsilon\left(B_n B_n^*\right) \right| \\
        &\quad + \left| \frac{1}{n} \tr g_\varepsilon\left(B_n B_n^*\right) - \frac{1}{2\pi} \int_0^{2\pi} \tr g_\varepsilon\left(\mtxfn{b}(x)\mtxfn{b}(x)^*\right) \,\dd x \right| \\
        &\quad + \left| \frac{1}{2\pi} \int_0^{2\pi} \tr g_\varepsilon\left(\mtxfn{b}(x)\mtxfn{b}(x)^*\right) \,\dd x - \frac{1}{2\pi} \int_0^{2\pi} \tr g\left(\mtxfn{b}(x)\mtxfn{b}(x)^*\right) \,\dd x \right| \\
        &< d\varepsilon + \varepsilon + d\varepsilon = (2d + 1)\varepsilon.
    \end{align*}
Since this holds for any $\varepsilon > 0$ and $n\ge N_{\ep}$, 
 \eqref{eqn:for-renyi-alpha-z} follows.
\end{proof}

\begin{rem}
Theorem \ref{thm:szego-block-poly} combined with polynomial approximation can be used to prove various other Szeg\H o-type limit theorems, e.g. of the form
\begin{align}\label{eq:Szego for infty}
\lim_{n\to\infty}\frac{1}{n}\tr g\bz f^{(1)}(A^{(1)}_n)+\ldots+f^{(r)}(A^{(r)}_n)\jz
=
\frac{1}{2\pi}\int_0^{2\pi}\tr g\bz f^{(1)}(\mtxsymfn{a}^{(1)}(x))+\ldots+f^{(r)}(\mtxsymfn{a}^{(r)}(x))\jz \,\dd x,
\end{align}
or
\begin{align}\label{eq:Szego for max}
\lim_{n\to\infty}\frac{1}{n}\tr g_2\bz C_n g_1(B_nB_n^*)C_n^*)\jz
=
\frac{1}{2\pi}\int_0^{2\pi}\tr g_2\bz \mtxfn{c}(x)g_1(\mtxfn{b}(x)\mtxfn{b}(x)^*)\mtxfn{c}(x)^*\jz\,\dd x,
\end{align}
with $B_n$ and $\mtxfn{b}$ as in \eqref{eq:B def}, and 
    \begin{align*}
    	C_n := \prod_{k=r+1}^m f^{(k)}(A^{(k)}_n),\ds\ds n\in\bN,
    	\ds\ds\ds\ds\ds
    	\mtxfn{c}(x) := \prod_{k=r+1}^m f^{(k)}(\mtxsymfn{a}^{(k)}(x)),\ds\ds x\in\bT,
    \end{align*}
provided that all expressions are well defined, and all functions can be uniformly approximated by polynomials 
on the relevant domains. 
While it does not seem clear how to cast all such Szeg\H o-type theorems in a universal form, and therefore 
provide a single proof that would cover all such statements, the proof for each particular case goes the same way as in the proofs of Theorem \ref{thm:szego-block-poly} and Corollary \ref{cor:for-renyi-alpha-z}, with
obvious adaptations. Therefore, we omit the proofs of \eqref{eq:Szego for infty}
and \eqref{eq:Szego for max}, and leave it to the reader to verify that they are valid in the particular settings where we apply them later.
\end{rem}

\section{Rényi divergences and error exponents for quasi-free states}
\label{sec:Renyi}

\subsection{Finite dimension}
\label{sec:finitedim}

For two positive definite operators $A,B\in\B(\hil)\pp$ and $\alpha\in\bR$, let
\begin{align}
\req_{\alpha,z}^{\op}(A\|B)&:=\bz A^{\frac{\alpha}{2z}}B^{\frac{1-\alpha}{z}}A^{\frac{\alpha}{2z}}\jz^z,
\ds\ds\ds z\in(0,+\infty),\\
\req_{\alpha,+\infty}^{\op}(A\|B)&:=\lim_{z\to+\infty}\req^{\op}_{\alpha,z}(A\|B)
=e^{\alpha\log A+(1-\alpha)\log B},\\
\req^{\op}_{\alpha,\geo}(A\|B)&:=B^{1/2}\bz B^{-1/2}AB^{-1/2}\jz^{\alpha}B^{1/2}
=A^{1/2}\bz A^{-1/2}BA^{-1/2}\jz^{1-\alpha}A^{1/2},\label{eq:maxRenyi def}
\end{align}
and for every $\tp\in(0,+\infty]\cup\{\geo\}$, let 
\begin{align*}
\req_{\alpha,\tp}(A\|B)&:=\Tr\req^{\op}_{\alpha,\tp}(A\|B),\\
\psi_{\alpha,\tp}(A\|B)&:=\log \req_{\alpha,\tp}(A\|B).
\end{align*}
The equality in \eqref{eq:maxRenyi def} is well known and easy to verify; see, e.g., 
\cite[Lemma 2.1]{HiaiMosonyi2017}.

From these, we define two different notions of R\'enyi $\alpha$-divergences. The first one
is inspired by \cite{Farooq_etal2024}, and is defined for every $\alpha\in\bR\setminus\{0,1\}$ and 
$\tp\in(0,+\infty]\cup\{\geo\}$ as
\begin{align*}
\redi_{\alpha,\tp}(A\|B):=\frac{1}{\oll\alpha-1}\psi_{\alpha,\tp}\bz\frac{A}{\Tr A}\Big\|\frac{B}{\Tr B}\jz,
\end{align*}
where
\begin{align*}
\oll\alpha:=\max\{\alpha,1-\alpha\},\ds\ds\ds\alpha\in\bR\setminus\{0,1\}.
\end{align*}
The second one, with the more common normalization, is defined for 
every $\alpha\in(0,1)\cup(1,+\infty)$ and $\tp\in(0,+\infty]\cup\{\geo\}$ as 
\begin{align}
\red_{\alpha,\tp}(A\|B):=\frac{1}{\alpha-1}\psi_{\alpha,\tp}(A\|B)-\frac{1}{\alpha-1}\log\Tr A
=
\frac{\oll\alpha-1}{\alpha-1}\redi_{\alpha,\gamma}(A\|B)+\log\Tr A-\log\Tr B.
\label{eq:Renyi scaling2}
\end{align}


For any $\gamma\in(0,+\infty]\cup\{\geo\}$, the above definitions can be extended to 
pairs of non-zero PSD operators $\rho,\sigma\in\B(\hil)\pne$ by 
\begin{align*}
\Delta_{\alpha,\gamma}(\rho\|\sigma):=\lim_{\ep\searrow 0}\Delta_{\alpha,\gamma}(\rho+\ep I\|\sigma+\ep I),
\end{align*}
where $\Delta$ may stand for $\req,\psi,\red$ or $\redi$.
This definition is consistent in the sense that the above definition becomes an identity for pairs of positive definite operators. 

It is straightforward to verify that for any 
$\alpha\in\bR\setminus\{0,1\}$
and any $\gamma\in(0,+\infty]\cup\{\geo\}$, 
$\req_{\alpha,\tp}$ is a 
quantum R\'enyi $\alpha$-quantity in the following sense: 
for any two commuting operators $\rho,\sigma$, both diagonal in some orthonormal basis 
$(\ket{\omega})_{\omega\in\Omega}$
as $\rho=\sum_{\omega\in\Omega}\tilde\rho(\omega)\pr{\omega}$, 
$\sigma=\sum_{\omega\in\Omega}\tilde\sigma(\omega)\pr{\omega}$, 
\begin{align*}
\req_{\alpha,\gamma}(\rho\|\sigma)&=\req_{\alpha}\bz\tilde\rho\|\tilde\sigma\jz:=
\lim_{\ep\searrow 0}\sum_{\omega\in\Omega}(\tilde\rho(\omega)+\ep)^{\alpha}(\tilde\sigma(\omega)+\ep)^{1-\alpha},\ds\ds\ds\gamma\in(0,+\infty]\cup\{\geo\},
\end{align*}
where $\req_{\alpha}(\tilde\rho\|\tilde\sigma)$ is the classical R\'enyi $\alpha$-quantity
\cite{Renyi} of the non-negative functions $\tilde\rho$ and $\tilde\sigma$.
Likewise,
$\red_{\alpha,\tp}$ is a quantum R\'enyi $\alpha$-divergence for every $\alpha\in(0,1)\cup(1,+\infty)$, and  
$\redi_{\alpha,\tp}$ is a quantum R\'enyi $\alpha$-divergence for every $\alpha\in\bR\setminus\{0,1\}$, 
with different normalizations.

We will call both 
$\red_{\alpha,z}$ and $\redi_{\alpha,z}$ \ki{R\'enyi $(\alpha,z)$-divergences} \cite{AD,JOPP}, 
$\red_{\alpha,+\infty}$ and $\redi_{\alpha,+\infty}$ \ki{log-Euclidean R\'enyi divergences} \cite{MO-cqconv}, 
and $\red_{\alpha,\geo}$ and $\redi_{\alpha,\geo}$ \ki{geometric R\'enyi divergences} \cite{Matsumoto_newfdiv,HiaiMosonyi2017}.
The special cases $\red_{\alpha,1}$ and $\redi_{\alpha,1}$ are called \ki{Petz-type} 
R\'enyi divergences \cite{P86}, and 
$\red_{\alpha,\alpha}$ and $\redi_{\alpha,\alpha}$
sandwiched R\' enyi divergences \cite{Renyi_new,WWY}.

A straightforward computation shows that 
\begin{align}
\redi_{1,\geo}(\rho\|\sigma)
&:=
\lim_{\alpha\to 1}\redi_{\alpha,\geo}(\rho\|\sigma)
=
D_{\geo}\bz\frac{\rho}{\Tr \rho}\Big\|\frac{\sigma}{\Tr \sigma}\jz
=
\frac{1}{\Tr \rho}D_{\geo}(\rho\|\sigma)-\log\Tr \rho+\log\Tr \sigma,
\label{eq:BS def1}\\
&=
\lim_{\alpha\to 1}\red_{\alpha,\geo}(\rho\|\sigma)-\log\Tr \rho+\log\Tr \sigma
=:
\red_{1,\geo}(\rho\|\sigma)-\log\Tr \rho+\log\Tr \sigma,
\end{align}
where 
\begin{align}
D_{\geo}(\rho\|\sigma)
&:=
\Tr \rho\log (\rho^{1/2}\sigma\inv \rho^{1/2})\label{eq:BS def0}\\
&=
\Tr \sigma^{1/2}\rho \sigma^{-1/2}\log (\sigma^{-1/2}\rho \sigma^{-1/2}),
\label{eq:BS def2}
\end{align}
when $\rho^0\le\sigma^0$ and $D_{\geo}(\rho\|\sigma):=+\infty$ otherwise,
is the \ki{Belavkin-Staszewski relative entropy}
\cite{BS} of $\rho$ and $\sigma$.
It is also known \cite{LinTomamichel15,Mosonyi_Hiai_Renyi_cont} that for any function
$(1-\delta,1+\delta)\ni\alpha\mapsto z(\alpha)\in(0,+\infty]$ with $\liminf_{\alpha\to 1}z(\alpha)>0$, 
\begin{align}
\lim_{\alpha\to 1}\redi_{\alpha,z(\alpha)}(\rho\|\sigma)
=
D\bz\frac{\rho}{\Tr \rho}\Big\|\frac{\sigma}{\Tr \sigma}\jz
&=
\frac{1}{\Tr \rho}D(\rho\|\sigma)-\log\Tr \rho+\log\Tr \sigma\label{eq:Umegaki def}\\
&=
\lim_{\alpha\to 1}\red_{\alpha,z(\alpha)}(\rho\|\sigma)-\log\Tr \rho+\log\Tr \sigma
\end{align}
where 
\begin{align}
D(\rho\|\sigma):=\Tr \rho(\log \rho-\log \sigma),\label{eq:Umegaki def2}
\end{align}
when $\rho^0\le\sigma^0$ and $D(\rho\|\sigma):=+\infty$ otherwise,
is the \ki{Umegaki relative entropy} \cite{Umegaki} of $\rho$ and $\sigma$.
The easily verifiable symmetry relation 
\begin{align}\label{eq:Renyi symmetry}
\redi_{\alpha,\gamma}(\rho\|\sigma)=\redi_{1-\alpha,\gamma}(\sigma\|\rho),\ds\ds\ds
\alpha\in\bR\setminus\{0,1\},\ds\gamma\in(0,+\infty]\cup\{\geo\},
\end{align}
then yields
\begin{align*}
\redi_{0,\geo}(\rho\|\sigma):=\lim_{\alpha\to 0}\redi_{\alpha,\geo}(\rho\|\sigma)
&=
D_{\geo}\bz\frac{\sigma}{\Tr \sigma}\Big\|\frac{\rho}{\Tr \rho}\jz,
\end{align*}
and for any function
$(-\delta,\delta)\ni\alpha\mapsto z(\alpha)\in(0,+\infty]$ with $\liminf_{\alpha\to 0}z(\alpha)>0$,
\begin{align*}
\lim_{\alpha\to 0}\redi_{\alpha,z(\alpha)}(\rho\|\sigma)
&=
D\bz\frac{\sigma}{\Tr \sigma}\Big\|\frac{\rho}{\Tr \rho}\jz.
\end{align*}
We will thus use the notations 
\begin{align*}
\redi_{1,z}(\rho\|\sigma)&:=D\bz\frac{\rho}{\Tr \rho}\Big\|\frac{\sigma}{\Tr \sigma}\jz,
& &
\redi_{0,z}(\rho\|\sigma):=D\bz\frac{\sigma}{\Tr \sigma}\Big\|\frac{\rho}{\Tr \rho}\jz,\\
\red_{1,z}(\rho\|\sigma)&:=\frac{1}{\Tr\rho}D\bz\rho\|\sigma\jz,
\end{align*}
for any $z\in(0,+\infty]$.

Finally, we note that for any fixed $\rho,\sigma\in\B(\hil)\pne$ and any $\tp\in(0,+\infty]$,
$\alpha\mapsto \redi_{\alpha,\tp}(\rho\|\sigma)$ is monotone decreasing on 
$(-\infty,1/2]$ and monotone increasing on $[1/2,+\infty)$, 
while $\alpha\mapsto \red_{\alpha,\tp}(\rho\|\sigma)$
is monotone increasing on $(0,+\infty)$; see \cite{HiaiJencova_az}.
Moreover, $\redi_{\alpha,\tp}(\rho\|\sigma)\ge 0$
for every $\alpha\in\bR$ and  $\tp\in(0,+\infty]\cup\{\geo\}$, with equality if and only if 
$\rho/\Tr\rho=\sigma/\Tr\sigma$.

\smallskip

In Lemma \ref{lemma:finite-dim psi} below, we give explicit formulas for the above R\'enyi divergences of two quasi-free states in terms of their symbols. For $D_{\alpha,+\infty}$, we will use the following simple identities.

\begin{lemma}\label{lemma:Fock exponentials}
Let $A_1,\ldots, A_r$ be positive definite operators on a finite-dimensional Hilbert space $\kil$ and 
$\alpha_1,\ldots,\alpha_r\in\bR$. Then 
\begin{align}
\exp\bz\sum_{i=1}^r\alpha_i \log A_i^{\wedge k}\jz
=
\bz\exp\bz\sum_{i=1}^r\alpha_i \log A_i\jz\jz^{\wedge k}
\label{eq:wedge exponential1}
\end{align}
for any $k=1,\ldots,\dim\kil$, and
\begin{align}
&\exp\bz\sum_{i=1}^r\alpha_i \log \Gamma(A_i)\jz
=
\Gamma\bz\exp\bz\sum_{i=1}^r\alpha_i \log A_i\jz\jz.
\label{eq:wedge exponential2}
\end{align}
\end{lemma}
\begin{proof}
We have
\begin{align*}
\sum_{i=1}^r\alpha_i\log A_i^{\wedge k}
=\sum_{i=1}^r(\alpha_i\log A_i^{\otimes k})\big\vert_{\wedge^k\kil}
&=
\sum_{i=1}^r\left.\left(\sum_{j=1}^k(\alpha_i\log A_i)\otimes I_{[k]\setminus\{j\}}\right)\right\vert_{\wedge^k\kil}\\
&=
\left.\left(\sum_{j=1}^k\bz\sum_{i=1}^r\alpha_i\log A_i\jz\otimes I_{[k]\setminus\{j\}}\right)\right\vert_{\wedge^k\kil},
\end{align*}
where $X\otimes I_{[k]\setminus\{j\}}$ is the canonical embedding of $X\in\B(\kil)$ into the 
$j$-th tensor component of $\B(\kil)^{\otimes k}$.
Taking the exponential of the first and the last expressions yields \eqref{eq:wedge exponential1}, and 
\eqref{eq:wedge exponential2} follows immediately.
\end{proof}

\begin{lemma}\label{lemma:finite-dim psi}
Let $\hil$ be a finite-dimensional Hilbert space and $Q,R\in\B(\hil)_{(0,1)}$
with corresponding quasi-free states $\omega_Q,\omega_R$ on $\car{\hil}$.
For any 
$\alpha\in\bR$ and any $\tp\in(0,+\infty]\cup\{\geo\}$, 
\begin{align}
\psi_{\alpha,\tp}(\what\omega_Q\|\what\omega_R)
&=
\psi_{\alpha,\tp}^{\qf}(Q\|R)\nn\\
&:=
\alpha\Tr\log(I - Q) + (1-\alpha)\Tr \log (I- R)
+\Tr\log \left(I + \req_{\alpha,\gamma}^{\op}\bz\frac{Q}{I-Q}\Big\|\frac{R}{I-R}\jz\right).
\label{eq:finite-dim psi}
\end{align}
\end{lemma}
\begin{proof}
According to \eqref{eq:quasifree density}, we can write $\what\omega_{Q} = \det(I - Q) \Gamma(\W{Q})$ 
and $\what\omega_{R} = \det(I - R) \Gamma(\W{R})$, where 
$\W{Q} = Q(I - Q)^{-1}$ and $\W{R} = R(I - R)^{-1}$.
    
Let us start with $\tp=z\in(0,+\infty)$. Then 
\begin{align}
\psi_{\alpha,z}\left(\what\omega_{Q}\|\what\omega_{R}\right)
&= 
\log \tr\left[\left(\what\omega_{Q}^{\frac{\alpha}{2z}} \what\omega_{R}^{\frac{1-\alpha}{z}} \what\omega_{Q}^{\frac{\alpha}{2z}}\right)^z\right]\nn\\
&=
\log \det(I - Q)^\alpha + \log \det(I - R)^{1-\alpha}+
\log\Tr  \bz\Gamma(\W{Q})^{\frac{\alpha}{2z}} \Gamma(\W{R})^{\frac{1-\alpha}{z}} \Gamma(\W{Q})^{\frac{\alpha}{2z}}\jz^z,\label{eq:finite-dim psi proof1}
\end{align}
and
\begin{align}
&\Tr  \bz\Gamma(\W{Q})^{\frac{\alpha}{2z}} \Gamma(\W{R})^{\frac{1-\alpha}{z}} \Gamma(\W{Q})^{\frac{\alpha}{2z}}\jz^z\nn\\
&\ds=
\Tr  \Gamma\bz \W{Q}^{\frac{\alpha}{2z}}\W{R}^{\frac{1-\alpha}{z}}\W{Q}^{\frac{\alpha}{2z}} \jz^z
=
\Tr  \Gamma\bz\req_{\alpha,z}^{\op}\bz\W{Q}\|\W{R}\jz\jz=
\det\left(I+ \req_{\alpha,z}^{\op}\bz\W{Q}\|\W{R}\jz\right),
\label{eq:finite-dim psi proof2}
\end{align}
where we used \eqref{eq:Fockop trace} in the last step.
Using also that $\log\det=\Tr\log$ on positive definite operators, 
\eqref{eq:finite-dim psi proof1}--\eqref{eq:finite-dim psi proof2} yield \eqref{eq:finite-dim psi}.
     
Next, let us consider $\tp=+\infty$. Then 
\begin{align}
\psi_{\alpha,+\infty}(\what\omega_{Q}\|\what\omega_{R})
=&
\log\Tr\exp\bz\alpha\log\what\omega_{Q}+(1-\alpha)\log\what\omega_{R}\jz\nn\\
=&
\log \det(I-Q)^\alpha +\log\det(I-R)^{1-\alpha}\nn\\
&+\log\Tr\exp\bz\alpha\log\Gamma(\W{Q})+(1-\alpha)\log\Gamma(\W{R})\jz,
\label{eq:finite-dim psi proof3}
\end{align}
and
\begin{align}
\Tr\exp\bz\alpha\log\Gamma(\W{Q})+(1-\alpha)\log\Gamma(\W{R})\jz
&=
\Tr\Gamma\bz\exp\bz\alpha\log\W{Q}+(1-\alpha)\log\W{R}\jz\jz\nn\\
&=
\det\bz I+\req_{\alpha,+\infty}^{\op}\bz\W{Q}\|\W{R}\jz\jz,
\label{eq:finite-dim psi proof4}
\end{align}
where the first equality is by 
\eqref{eq:wedge exponential2}, and the second equality is due to \eqref{eq:Fockop trace}.
Replacing again all $\log\det$ with $\Tr\log$, 
\eqref{eq:finite-dim psi proof3}--\eqref{eq:finite-dim psi proof4}
yield \eqref{eq:finite-dim psi}.

Finally, we consider the case $\tp=\geo$. Then 
\begin{align}
\psi_{\alpha,\geo}(\what\omega_{Q}\|\what\omega_{R})&=
\log\Tr\what\omega_{R}^{1/2}\bz\what\omega_{R}^{-1/2}\what\omega_{Q}\what\omega_{R}^{-1/2}\jz^{\alpha}\omega_{R}^{1/2}\nn\\
&=
\log \det(I - Q)^\alpha +  \log \det(I - R)^{1-\alpha}\nn\\
&\ds\ds+\log\Tr\Gamma(\W{R})^{1/2}\bz\Gamma(\W{R})^{-1/2}\Gamma(\W{Q})
\Gamma(\W{R})^{-1/2}\jz^{\alpha}\Gamma(\W{R})^{1/2},
\label{eq:finite-dim psi proof5}
\end{align}
and
\begin{align}
&\Tr\Gamma(\W{R})^{1/2}\bz\Gamma(\W{R})^{-1/2}\Gamma(\W{Q})
\Gamma(\W{R})^{-1/2}\jz^{\alpha}\Gamma(\W{R})^{1/2}\nn\\
&\ds=
\Tr\Gamma\bz\W{R}^{1/2}\bz\W{R}^{-1/2}\W{Q}^{1/2}\W{R}^{-1/2}\jz^{\alpha}\W{R}^{1/2}\jz\nn\\
&\ds=
\Tr\Gamma\bz\req_{\alpha,\geo}^{\op}(\W{Q}\|\W{R})\jz=
\det(I+\req_{\alpha,\geo}^{\op}(\W{Q}\|\W{R})).
\label{eq:finite-dim psi proof6}
\end{align}
Replacing again all $\log\det$ with $\Tr\log$, 
\eqref{eq:finite-dim psi proof5}--\eqref{eq:finite-dim psi proof6}
yield \eqref{eq:finite-dim psi}.
\end{proof}

For the Petz-type R\'enyi quantities one may obtain a different expression as follows.

\begin{lemma}
In the setting of Lemma \ref{lemma:finite-dim psi},
\begin{align*}
\psi_{\alpha,1}\left(\what\omega_{Q}\|\what\omega_{R}\right)
=
\log\det\left[ Q^{\alpha}R^{1-\alpha}+(I-Q)^{\alpha}(I-R)^{1-\alpha}\right],\ds\ds\ds
\alpha\in\bR.
\end{align*}
\end{lemma}
\begin{proof}
Follows by a straightforward computation, as
\begin{align*}
\psi_{\alpha,1}\left(\what\omega_{Q}\|\what\omega_{R}\right)
&= 
\log \tr\bz\what\omega_{Q}^{\alpha} \what\omega_{R}^{1-\alpha}\jz\\
&=
\log \det(I - Q)^\alpha + \log \det(I - R)^{1-\alpha}+
\log\Tr  \bz\Gamma(\W{Q})^{\alpha} \Gamma(\W{R})^{1-\alpha}\jz\\
&=
\log \det(I - Q)^\alpha + \log \det(I - R)^{1-\alpha}+
\log\Tr \Gamma\bz\W{Q}^{\alpha} \W{R}^{1-\alpha}\jz\\
&=
\log \det(I - Q)^\alpha + \log \det(I - R)^{1-\alpha}+
\log\det\bz I+\W{Q}^{\alpha}\W{R}^{1-\alpha}\jz\\
&=
\log\det\left[ Q^{\alpha}R^{1-\alpha}+(I-Q)^{\alpha}(I-R)^{1-\alpha}\right].
\end{align*}
\end{proof}

\begin{rem}
Note that $\psi_{\alpha,\tp}^{\qf}(Q\|R)$ in \eqref{eq:finite-dim psi} may be written as
\begin{align*}
\psi_{\alpha,\tp}^{\qf}(Q\|R)
&=
\log\underbrace{\left[\det(I-Q)^{\alpha}\det(I-R)^{1-\alpha}\det\bz I + \req_{\alpha,\gamma}^{\op}\bz\frac{Q}{I-Q}\Big\|\frac{R}{I-R}\jz\jz\right]}_{=:\req_{\alpha,\tp}^{\qf}(Q\|R)},
\end{align*}
where
\begin{align*}
\req_{\alpha,\tp}^{\qf}(Q\|R)
&=
\det\Bigg[
(I-R)^{\frac{1-\alpha}{2}}(I-Q)^{\frac{\alpha}{2}}
\Bigg[I+
\req_{\alpha,\gamma}^{\op}\bz\frac{Q}{I-Q}\Big\|\frac{R}{I-R}\jz\Bigg]
(I-Q)^{\frac{\alpha}{2}}(I-R)^{\frac{1-\alpha}{2}}\Bigg]\,.
\end{align*}
\end{rem}

\begin{rem}
The log-Euclidean R\'enyi divergence can be defined for multiple positive definite operators 
$A_1,\ldots,A_r\in\B(\kil)\pp$ and positive weights $\alpha_1,\ldots,\alpha_r$ summing to $1$, and expressed via a variational formula, as
\begin{align}\label{eq:multivariate log-Euclidean}
\psi_{\vecc{\alpha}}(A_1,\ldots,A_r)
:=\log\Tr\exp\bz\sum_{i=1}^r\alpha_i \log A_i\jz
=
-\min_{\omega\in\S(\kil)}\sum_{i=1}^rD(\omega\|A_i),
\end{align}
where the unique optimal $\omega$ is given by
\begin{align*}
\oll\omega_{\vecc{\alpha}}=\frac{\exp\bz\sum_{i=1}^r\alpha_i \log A_i\jz}{\Tr\exp\bz\sum_{i=1}^r\alpha_i \log A_i\jz};
\end{align*}
see \cite{mosonyi2022geometric,MO-cqconv} for details. It is worth noting that if $A_i=\what\omega_{Q_i}$ are quasi-free states with 
symbols $Q_i\in\B(\hil)_{(0,1)}$, where $\hil$ is finite-dimensional, then for 
any $\vecc{\alpha}$, $\oll\omega_{\vecc{\alpha}}$ is also quasi-free with symbol 
\begin{align*}
\oll Q=f\inv\bz\exp\bz\sum_{i=1}^r\alpha_i \log\frac{Q_i}{I-Q_i}\jz\jz=
I-\bz I+\exp\bz\sum_{i=1}^r\alpha_i \log\frac{Q_i}{I-Q_i}\jz\jz\inv,
\end{align*}
where $f(x):=x/(1-x)$, $x\in(0,1)$. Indeed, this follows immediately from the fact that by \eqref{eq:quasifree density} and Lemma \ref{lemma:Fock exponentials}, 
$\oll\omega_{\vecc{\alpha}}$ in this case is proportional to 
\begin{align*}
\exp\bz\sum_{i=1}^r\alpha_i\log\what\omega_{Q_i}\jz
&=
\exp\bz\sum_{i=1}^r\alpha_i\log\det(I-Q_i)+\sum_{i=1}^r\alpha_i\log\Gamma(f(Q_i))\jz\\
&=
\bz\prod_{i=1}^r\det(I-Q_i)^{\alpha_i}\jz\Gamma\bz\exp\bz\sum_{i=1}^r\alpha_i\log f(Q_i)\jz\jz.
\end{align*}
Equivalently, the multivariate R\'enyi $\psi$ quantity for quasi-free states can also be expressed
by a variational formula as in \eqref{eq:multivariate log-Euclidean}, but with the minimization restricted to quasi-free states.
\end{rem}

\begin{lemma}\label{lemma:psi derivative}
In the setting of Lemma \ref{lemma:finite-dim psi},
$\bR\times(0,+\infty)\ni(\alpha,z)\mapsto \psi_{\alpha,z}(\what\omega_Q\|\what\omega_R)$
is infinitely many times differentiable, and for any fixed $z\in(0,+\infty)$, 
\begin{align}
\partial_{\alpha}\psi_{\alpha,z}(\what\omega_Q\|\what\omega_R)
=&
\Tr\log(I - Q)-\Tr \log (I- R)\nn\\
&+\Tr\left[(I+W_{\alpha,z}^{-z})\inv \left[(\log\W{Q})
-
W_{\alpha,z}\inv
\W{Q}^{\frac{\alpha}{2z}}(\log\W{R})\W{R}^{\frac{1-\alpha}{z}}\W{Q}^{\frac{\alpha}{2z}}
\right]\right],
\label{eq:psi derivative}
\end{align}
where
\begin{align}\label{eq:Walphaz def}
W_{\alpha,z}
:=
\req_{\alpha,z}^{\op}\bz\frac{Q}{I-Q}\Big\|\frac{R}{I-R}\jz^{1/z}
=
\req_{\alpha,z}^{\op}\bz\W{Q}\|\W{R}\jz^{1/z}
=
\W{Q}^{\frac{\alpha}{2z}}\W{R}^{\frac{1-\alpha}{z}}\W{Q}^{\frac{\alpha}{2z}}.
\end{align}
\end{lemma}
\begin{proof}
Since $\psi_{\alpha,z}(\what\omega_Q\|\what\omega_R)$ is put together from functions that are all smooth (in fact, analytic) in $(\alpha,z)$, $\psi_{\alpha,z}(\what\omega_Q\|\what\omega_R)$ itself has this property. 
It is clear that the derivative of the first two terms in \eqref{eq:finite-dim psi} w.r.t.~$\alpha$ gives the first two terms in \eqref{eq:psi derivative}.
The derivative of the last term in \eqref{eq:finite-dim psi} can be computed 
using \eqref{eq:Tr derivative} as
\begin{align*}
&\partial_{\alpha}\Tr\log \left(I + \req_{\alpha,z}^{\op}\bz\frac{Q}{I-Q}\Big\|\frac{R}{I-R}\jz\right)\\
&\ds=\partial_{\alpha}\Tr\log \left(I + W_{\alpha,z}^z\right)
=
\Tr\left[(I+W_{\alpha,z}^z)\inv zW_{\alpha,z}^{z-1}\partial_{\alpha}W_{\alpha,z}
\right]\\
&\ds=
\Tr\left[(I+W_{\alpha,z}^z)\inv zW_{\alpha,z}^{z-1}
\left[\frac{1}{2z}(\log\W{Q})W_{\alpha,z}+\frac{1}{2z}W_{\alpha,z}(\log\W{Q})
-\frac{1}{z}\W{Q}^{\frac{\alpha}{2z}}(\log\W{R})\W{R}^{\frac{1-\alpha}{z}}\W{Q}^{\frac{\alpha}{2z}}
\right]\right]\\
&\ds=
\Tr\left[(I+W_{\alpha,z}^z)\inv W_{\alpha,z}^z(\log\W{Q})\right]
-
\Tr\left[(I+W_{\alpha,z}^z)\inv W_{\alpha,z}^{z-1}
\W{Q}^{\frac{\alpha}{2z}}(\log\W{R})\W{R}^{\frac{1-\alpha}{z}}\W{Q}^{\frac{\alpha}{2z}}
\right]\\
&\ds=
\Tr\left[(I+W_{\alpha,z}^{-z})\inv \left[(\log\W{Q})
-
W_{\alpha,z}\inv
\W{Q}^{\frac{\alpha}{2z}}(\log\W{R})\W{R}^{\frac{1-\alpha}{z}}\W{Q}^{\frac{\alpha}{2z}}
\right]\right],
\end{align*}
completing the proof of \eqref{eq:psi derivative}.
\end{proof}

\begin{lemma}\label{lemma:psi infty derivative}
In the setting of Lemma \ref{lemma:finite-dim psi},
$\bR\ni\alpha\mapsto \psi_{\alpha,+\infty}(\what\omega_Q\|\what\omega_R)$
is infinitely many times differentiable, and 
\begin{align}
\partial_{\alpha}\psi_{\alpha,+\infty}(\what\omega_Q\|\what\omega_R)
=&
\Tr\log(I - Q)-\Tr \log (I- R)\nn\\
&+\Tr\left[
\left[
I+\exp\bz-\alpha\log\W{Q}-(1-\alpha)\log\W{R}\jz\right]\inv\left[\log\W{Q}-\log\W{R}
\right]
\right].
\label{eq:psi infty derivative}
\end{align}
\end{lemma}
\begin{proof}
The claim about infinite differentiability is obvious.
It is clear that the derivative of the first two terms in \eqref{eq:finite-dim psi} w.r.t.~$\alpha$ gives the first two terms in \eqref{eq:psi infty derivative}.
The derivative of the last term in \eqref{eq:finite-dim psi} can be computed 
using \eqref{eq:Tr derivative} as
\begin{align*}
&\partial_{\alpha}\Tr\log\bz I+\req_{\alpha,+\infty}^{\op}(\W{Q}\|\W{R})\jz\\
&\ds=
\Tr\bz I+\req_{\alpha,+\infty}^{\op}(\W{Q}\|\W{R})\jz\inv\req_{\alpha,+\infty}^{\op}(\W{Q}\|\W{R})
\left[\log\W{Q}-\log\W{R}\right]\\
&\ds=
\Tr\bz I+\req_{\alpha,+\infty}^{\op}(\W{Q}\|\W{R})\inv\jz\inv\left[\log\W{Q}-\log\W{R}\right],
\end{align*}
which is exactly the last term in \eqref{eq:psi infty derivative}.
\end{proof}

\begin{lemma}\label{lemma:psi max derivative}
In the setting of Lemma \ref{lemma:finite-dim psi},
$\bR\ni\alpha\mapsto \psi_{\alpha,\geo}(\what\omega_Q\|\what\omega_R)$
is infinitely many times differentiable, and 
\begin{align}
\partial_{\alpha}\psi_{\alpha,\geo}(\what\omega_Q\|\what\omega_R)
=&
\Tr\log(I - Q)-\Tr \log (I- R)\nn\\
&+\Tr\bz I+\req_{\alpha,\geo}^{\op}(\W{Q}\|\W{R})\inv\jz\inv
\W{R}^{-1/2}\bz\log\bz\W{R}^{-1/2}\W{Q}\W{R}^{-1/2}\jz\jz\W{R}^{1/2}
\label{eq:psi max derivative1}\\
=&
\Tr\log(I - Q)-\Tr \log (I- R)\nn\\
&+\Tr\bz I+\req_{\alpha,\geo}^{\op}(\W{Q}\|\W{R})\inv\jz\inv
\W{Q}^{-1/2}\bz\log\bz\W{Q}^{1/2}\W{R}\inv\W{Q}^{1/2}\jz\jz\W{Q}^{1/2}.
\label{eq:psi max derivative2}
\end{align}
\end{lemma}
\begin{proof}
The claim about infinite differentiability is obvious.
It is clear that the derivative of the first two terms in \eqref{eq:finite-dim psi} w.r.t.~$\alpha$ gives the first two terms in \eqref{eq:psi max derivative1} as well as in \eqref{eq:psi max derivative2}.
The derivative of the last term in \eqref{eq:finite-dim psi} can be computed 
using \eqref{eq:Tr derivative} as
\begin{align*}
&\partial_{\alpha}\Tr\log\bz I+\req_{\alpha,\geo}^{\op}(\W{Q}\|\W{R})\jz\\
&\ds=
\Tr\bz I+\req_{\alpha,\geo}^{\op}(\W{Q}\|\W{R})\jz\inv
\W{R}^{1/2}\left[\partial_{\alpha}\bz\W{R}^{-1/2}\W{Q}\W{R}^{-1/2}\jz^{\alpha}\right]\W{R}^{1/2}\\
&\ds=
\Tr\bz I+\req_{\alpha,\geo}^{\op}(\W{Q}\|\W{R})\jz\inv
\W{R}^{1/2}\left[\bz\W{R}^{-1/2}\W{Q}\W{R}^{-1/2}\jz^{\alpha}\log\bz\W{R}^{-1/2}\W{Q}\W{R}^{-1/2}\jz\right]\W{R}^{1/2}\\
&\ds=
\Tr\bz I+\req_{\alpha,\geo}^{\op}(\W{Q}\|\W{R})\jz\inv
\req_{\alpha,\geo}^{\op}(\W{Q}\|\W{R})\W{R}^{-1/2}\bz\log\bz\W{R}^{-1/2}\W{Q}\W{R}^{-1/2}\jz\jz\W{R}^{1/2}\\
&\ds=
\Tr\bz I+\req_{\alpha,\geo}^{\op}(\W{Q}\|\W{R})\inv\jz\inv
\W{R}^{-1/2}\bz\log\bz\W{R}^{-1/2}\W{Q}\W{R}^{-1/2}\jz\jz\W{R}^{1/2},
\end{align*}
which is exactly the last term in \eqref{eq:psi max derivative1}.

Alternatively, one may use the second expression for $\req_{\alpha,\geo}^{\op}$ in 
\eqref{eq:maxRenyi def} to obtain
\begin{align*}
&\partial_{\alpha}\Tr\log\bz I+\req_{\alpha,\geo}^{\op}(\W{Q}\|\W{R})\jz\\
&\ds=
\Tr\bz I+\req_{\alpha,\geo}^{\op}(\W{Q}\|\W{R})\jz\inv
\W{Q}^{1/2}\left[\partial_{\alpha}\bz\W{Q}^{-1/2}\W{R}\W{Q}^{-1/2}\jz^{1-\alpha}\right]\W{Q}^{1/2}\\
&\ds=
-\Tr\bz I+\req_{\alpha,\geo}^{\op}(\W{Q}\|\W{R})\jz\inv
\W{Q}^{1/2}\left[\bz\W{Q}^{-1/2}\W{R}\W{Q}^{-1/2}\jz^{1-\alpha}\log\bz\W{Q}^{-1/2}\W{R}\W{Q}^{-1/2}\jz\right]\W{Q}^{1/2}\\
&\ds=
-\Tr\bz I+\req_{\alpha,\geo}^{\op}(\W{Q}\|\W{R})\jz\inv
\req_{\alpha,\geo}^{\op}(\W{Q}\|\W{R})\W{Q}^{-1/2}\bz\log\bz\W{Q}^{-1/2}\W{R}\W{Q}^{-1/2}\jz\jz\W{Q}^{1/2}\\
&\ds=
-\Tr\bz I+\req_{\alpha,\geo}^{\op}(\W{Q}\|\W{R})\inv\jz\inv
\W{Q}^{-1/2}\bz\log\bz\W{Q}^{-1/2}\W{R}\W{Q}^{-1/2}\jz\jz\W{Q}^{1/2}\\
&\ds=
\Tr\bz I+\req_{\alpha,\geo}^{\op}(\W{Q}\|\W{R})\inv\jz\inv
\W{Q}^{-1/2}\bz\log\bz\W{Q}^{1/2}\W{R}\inv\W{Q}^{1/2}\jz\jz\W{Q}^{1/2},
\end{align*}
which is exactly the last term in \eqref{eq:psi max derivative2}.
\end{proof}

\begin{cor}
In the setting of Lemma \ref{lemma:finite-dim psi},
\begin{align}
D(\what\omega_Q\|\what\omega_R)
=&
D^{\qf}(Q\|R)\nn\\
:=&
\Tr \left[Q\bz\log Q-\log R\jz
+(I-Q)\bz\log(I-Q)-\log (I-R)\jz\right],
\label{eq:qf Umegaki}\\
D_{\geo}(\what\omega_Q\|\what\omega_R)
=&
D_{\geo}^{\qf}(Q\|R)\nn\\
:=&
\Tr\log(I - Q)-\Tr \log (I- R)\nn\\
&+\Tr Q
\W{R}^{-1/2}\bz\log\bz\W{R}^{-1/2}\W{Q}\W{R}^{-1/2}\jz\jz\W{R}^{1/2}
\label{eq:qf BS1}\\
=&
\Tr\log(I - Q)-\Tr \log (I- R)+\Tr Q\log\bz\W{Q}^{1/2}\W{R}\inv\W{Q}^{1/2}\jz.
\label{eq:qf BS2}
\end{align}
\end{cor}
\begin{proof}
We have
\begin{align*}
D(\what\omega_Q\|\what\omega_R)
=&
\lim_{\alpha\to 1}D_{\alpha,+\infty}(\what\omega_Q\|\what\omega_R)
=
\partial_{\alpha}\psi_{\alpha,+\infty}(\what\omega_Q\|\what\omega_R)\big\vert_{\alpha=1}\\
=&
\Tr\log(I - Q)-\Tr \log (I- R)\nn\\
&+\Tr\left[
\left[
I+\exp\bz-\alpha\log\W{Q}-(1-\alpha)\log\W{R}\jz\right]\inv\left[\log\W{Q}-\log\W{R}
\right]
\right]\Big\vert_{\alpha=1}\\
=&
\Tr\log(I - Q)-\Tr \log (I- R)\nn\\
&+\Tr\left[
\left[I+\W{Q}\inv\right]\inv\left[\log\W{Q}-\log\W{R}
\right]
\right]\\
=&
\Tr\log(I - Q)-\Tr \log (I- R)
+\Tr Q\left[\log\frac{Q}{I-Q}-\log\frac{R}{I-R}\right]\\
=&\Tr \left[Q\bz\log Q-\log R\jz
+(I-Q)\bz\log(I-Q)-\log (I-R)\jz\right],
\end{align*}
where the first equality follows from \eqref{eq:Umegaki def} with $z(\alpha)\equiv+\infty$,
the second equality is by definition, 
the third equality follows from \eqref{eq:psi infty derivative} and the rest are obvious.
This proves \eqref{eq:qf Umegaki}.

Similarly,
\begin{align*}
D_{\geo}(\what\omega_Q\|\what\omega_R)
=&
\lim_{\alpha\to 1}D_{\alpha,\geo}(\what\omega_Q\|\what\omega_R)
=
\partial_{\alpha}\psi_{\alpha,\geo}(\what\omega_Q\|\what\omega_R)\big\vert_{\alpha=1}\\
=&
\Tr\log(I - Q)-\Tr \log (I- R)\nn\\
&+\Tr\underbrace{\bz I+\req_{1,\geo}^{\op}(\W{Q}\|\W{R})\inv\jz\inv}_{=Q}
\W{R}^{-1/2}\bz\log\bz\W{R}^{-1/2}\W{Q}\W{R}^{-1/2}\jz\jz\W{R}^{1/2}
\end{align*}
where the first equality follows from \eqref{eq:BS def1}--\eqref{eq:BS def2},
the second equality is by definition, 
the third equality follows from \eqref{eq:psi max derivative1} and the rest are obvious.
This proves \eqref{eq:qf BS1}. Using instead \eqref{eq:psi max derivative2} in the third equality, we get 
\begin{align*}
D_{geo}(\what\omega_Q\|\what\omega_R)
=&
\Tr\log(I - Q)-\Tr \log (I- R)\nn\\
&+\Tr\bz I+\req_{1,\geo}^{\op}(\W{Q}\|\W{R})\inv\jz\inv
\W{Q}^{-1/2}\bz\log\bz\W{Q}^{1/2}\W{R}\inv\W{Q}^{1/2}\jz\jz\W{Q}^{1/2}\\
=&
\Tr\log(I - Q)-\Tr \log (I- R)+\Tr Q\log\bz\W{Q}^{1/2}\W{R}\inv\W{Q}^{1/2}\jz,
\end{align*}
proving \eqref{eq:qf BS2}.
\end{proof}

\subsection{Regularized Rényi divergences of translation-invariant states} \label{sec:renyi-divs-qf}
\label{sec:regularized Renyi}

\subsubsection{$(\alpha,z)$ and geometric}
\label{sec:regularized az}

Let us consider a doubly infinite fermion chain with $d$ internal degrees of freedom, 
described by the single-particle Hilbert space $\ell^2_d(\bZ)$. 
For the rest of the section, 
let $\omega_Q$ and $\omega_R$ be 
translation-invariant quasi-free states of this system 
with corresponding symbol functions
    $\mtxsymfn{q},\mtxsymfn{r}\in L^{\infty}_{d\times d}(\bT)$, i.e., 
    $Q = \ft^{-1} M_{\mtxsymfn{q}} \ft$ and $R = \ft^{-1} M_{\mtxsymfn{r}} \ft$.
We will assume throughout this section that there exists a constant $c \in (0, 1/2)$  such that 
\begin{align}\label{eq:strict bound1}
c I_d \le \mtxsymfn{q}(x), \mtxsymfn{r}(x) \le (1-c)I_d
\end{align}
almost everywhere, or equivalently, that 
\begin{align}\label{eq:strict bound2}
cI\le Q,R\le (1-c)I,
\end{align}
and we define
\begin{align*}
\weightmtxfn{Q}(x) := \mtxsymfn{q}(x)(I_d - \mtxsymfn{q}(x))^{-1},
\ds\ds\ds
\weightmtxfn{R}(x) := \mtxsymfn{r}(x)(I_d - \mtxsymfn{r}(x))^{-1},\ds\ds\ds
x\in\bT.
\end{align*}
We may and will assume without loss of generality that \eqref{eq:strict bound1} holds at every $x\in\bT$.
Assumptions \eqref{eq:strict bound1}--\eqref{eq:strict bound2} guarantee that 
\begin{align}\label{eq:strict bound3}
cI_{nd}\le Q_n,R_n\le(1-c)I_{nd},\ds n\in\bN,\ds\ds\ds
\frac{c}{1-c}I_d\le \weightmtxfn{Q}(x) ,\weightmtxfn{R}(x) \le\frac{1-c}{c} I_d,\ds x\in\bT,
\end{align}
where $Q_n$ is as in \eqref{eq:cutoff def}, and $R_n$ is defined analogously.

We are interested in the \ki{regularized} versions of the R\'enyi quantities, defined as
\begin{align*}
\divv^{\reg}(\omega_Q\|\omega_R):=\lim_{n\to+\infty}\frac{1}{n}\divv(\what\omega_{Q_n}\|\what\omega_{R_n})
\end{align*}
whenever the limit exists, 
where $\divv=\psi_{\alpha,\gamma}$ or $\divv=\redi_{\alpha,\gamma}$ for some 
$\alpha\in\bR$ and $\gamma\in(0,+\infty]\cup\{\geo\}$.

\begin{theorem}\label{thm:psi-existence}  
For every $\alpha\in\bR$ and $\tp\in(0,+\infty]\cup\{\geo\}$,    
\begin{align}
&\psi^{\reg}_{\alpha,\tp}(\omega_Q\|\omega_R)
=
\frac{1}{2\pi} \int_0^{2\pi}\psi_{\alpha,\tp}(\what\omega_{\mtxsymfn{q}(x)}\|\what\omega_{\mtxsymfn{r}(x)})\,\dd x
=
\frac{1}{2\pi} \int_0^{2\pi}\psi_{\alpha,\tp}^{\qf}(\mtxsymfn{q}(x)\|\mtxsymfn{r}(x))\,\dd x
\label{eq:regularized_psi1}
\\
&= \frac{1}{2\pi} \int_0^{2\pi} \tr \Big[\log(I_d - \mtxsymfn{q}(x))^{\alpha} 
+ \log(I_d - \mtxsymfn{r}(x))^{1-\alpha} + \log\bz I_d + 
\req^{\op}_{\alpha,\tp}(\weightmtxfn{Q}(x)\| \weightmtxfn{R}(x))\jz\Big] \,\dd x.
\label{eq:regularized_psi2}
\end{align}
\end{theorem}
\begin{proof}
The equalities of the two integrals in \eqref{eq:regularized_psi1} and the one in \eqref{eq:regularized_psi2} 
are by definition.
By Lemma \ref{lemma:finite-dim psi},
\begin{align}
\psi^{\reg}_{\alpha,\tp}(\omega_Q\|\omega_R)
=&
\alpha\lim_{n\to+\infty}\frac{1}{n}\Tr\log(I_{nd} - Q_n) 
+(1-\alpha)\lim_{n\to+\infty}\frac{1}{n}\Tr \log (I_{nd}- R_n)
\label{eq:regularized_psi_proof1}\\
&+\lim_{n\to+\infty}\frac{1}{n}\Tr\log \left(I_{nd} + \req_{\alpha,\gamma}^{\op}\bz\frac{Q_n}{I_{nd}-Q_n}\Big\|\frac{R_n}{I_{nd}-R_n}\jz\right).
\label{eq:regularized_psi_proof2}
\end{align}

By assumption \eqref{eq:strict bound2}, $\log$ is continuous on 
$\conv(\spec(I-Q))\cup\conv(\spec(I-R))$, and hence, by Theorem \ref{thm:szego-block-poly}, the limits in 
\eqref{eq:regularized_psi_proof1} are equal to 
\begin{align}
\frac{1}{2\pi} \int_0^{2\pi} \tr \left[\alpha\log(I_d - \mtxsymfn{q}(x))
+(1-\alpha) \log(I_d - \mtxsymfn{r}(x))\right] \,\dd x,
\end{align}
which gives the first two terms in the integral in \eqref{eq:regularized_psi2}.
Hence, our aim is to show that the limit in \eqref{eq:regularized_psi_proof2}
is equal to the last integral term in \eqref{eq:regularized_psi2}.

Let us first consider $\tp=z\in(0,+\infty)$. Define
$g(x) := \log(1 + x^z)$, let $B_n$ and $\mtxfn{b}$ be as in \eqref{eq:B def} with 
$f^{(1)}(x) = (x/(1-x))^{\frac{\alpha}{2z}}$, $f^{(2)}(x) = (x/(1-x))^{\frac{1-\alpha}{2z}}$, and 
$\mtxsymfn{a}^{(1)}=\mtxsymfn{q}$,
$\mtxsymfn{a}^{(2)}=\mtxsymfn{r}$, so that 
$\req_{\alpha,\gamma}^{\op}\bz\frac{Q_n}{I_{nd}-Q_n}\Big\|\frac{R_n}{I_{nd}-R_n}\jz=B_nB_n^*$. 
Then the limit in \eqref{eq:regularized_psi_proof2} is equal to 
\begin{align}
\lim_{n\to+\infty}\frac{1}{n} \tr g\left(B_n B_n^*\right)
&=
\frac{1}{2\pi} \int_0^{2\pi} \tr g\left(\mtxfn{b}(x)\mtxfn{b}(x)^*\right) \,\dd x\nn\\
&=
\frac{1}{2\pi} \int_0^{2\pi} \tr \log\Big(I_d + 
\underbrace{\bz(\weightmtxfn{Q}(x))^{\frac{\alpha}{2z}} (\weightmtxfn{R}(x))^{\frac{1-\alpha}{z}}(\weightmtxfn{Q}(x))^{\frac{\alpha}{2z}}\jz^{z}}_{=\req^{\op}_{\alpha,z}(\weightmtxfn{Q}(x)\| \weightmtxfn{R}(x))}\Big) \,\dd x
\label{eq:regularized az proof2}
\end{align}
where the first equality is by definition, and the second equality 
follows from Corollary \ref{cor:for-renyi-alpha-z}, since 
all the continuity requirements for $f^{(1)},f^{(2)}$ and $g$ in Corollary \ref{cor:for-renyi-alpha-z} are met 
due to the assumption in \eqref{eq:strict bound2}. 
This completes the proof of \eqref{eq:regularized_psi2} in the above case.

Next, let us consider $\tp=+\infty$. 
In this case the limit in \eqref{eq:regularized_psi_proof2} is equal to 
\begin{multline}
\lim_{n\to+\infty}\frac{1}{n}\Tr\log\bz I_{nd}+\exp\bz\alpha\log\W{Q_n}+
(1-\alpha)\log\W{R_n}\jz\jz\\
=
\frac{1}{2\pi} \int_0^{2\pi} 
\tr \log\big( I_d+\underbrace{\exp\bz\alpha\log\weightmtxfn{Q}(x)+(1-\alpha)\log \weightmtxfn{R}(x)\jz}_{
=\req^{\op}_{\alpha,+\infty}(\weightmtxfn{Q}(x)\| \weightmtxfn{R}(x))}\big)\,\dd x,
\end{multline}
where the second expression follows from the first 
as a special case of \eqref{eq:Szego for infty}.

Finally, consider the case $\tp=\geo$. Then 
the limit in \eqref{eq:regularized_psi_proof2} is equal to 
\begin{multline*}
\lim_{n\to+\infty}\frac{1}{n}\Tr\log\bz I_{nd}+\W{R_n}^{1/2}\bz\W{R_n}^{-1/2}\W{Q_n}\W{R_n}^{-1/2}\jz^{\alpha}\W{R_n}^{1/2}\jz\\
=
\frac{1}{2\pi} \int_0^{2\pi} 
\tr \log\big( I_d+\underbrace{
(\weightmtxfn{R}(x))^{1/2}\big((\weightmtxfn{R}(x))^{-1/2}\weightmtxfn{Q}(x)
(\weightmtxfn{R}(x))^{-1/2}\big)^{\alpha}(\weightmtxfn{R}(x))^{1/2}}_{
=\req^{\op}_{\alpha,\geo}(\weightmtxfn{Q}(x)\| \weightmtxfn{R}(x))}\big)\,\dd x,
\end{multline*}
where the second expression follows from the first 
as a special case of \eqref{eq:Szego for max}.
\end{proof}

Theorem \ref{thm:psi-existence} yields immediately the following:
\begin{theorem}\label{thm:psi-existence-2}  
For every $\alpha\in\bR\setminus\{0,1\}$ and $\tp\in(0,+\infty]\cup\{\geo\}$,    
\begin{align}
&\redi_{\alpha,\tp}^{\reg}(\omega_Q\|\omega_R)
=
\frac{1}{2\pi} \int_0^{2\pi}\redi_{\alpha,\tp}(\what\omega_{\mtxsymfn{q}(x)}\|\what\omega_{\mtxsymfn{r}(x)})\,\dd x
=
\frac{1}{2\pi} \int_0^{2\pi}\redi_{\alpha,\tp}^{\qf}(\mtxsymfn{q}(x)\|\mtxsymfn{r}(x))\,\dd x
\label{eq:regularized_psi1-2}
\\
&= \frac{1}{2\pi} \int_0^{2\pi} \frac{1}{\oll\alpha-1}\tr \Big[\log(I_d - \mtxsymfn{q}(x))^{\alpha} 
+ \log(I_d - \mtxsymfn{r}(x))^{1-\alpha} + \log\bz I_d + 
\req^{\op}_{\alpha,\tp}(\weightmtxfn{Q}(x)\| \weightmtxfn{R}(x))\jz\Big] \,\dd x.
\label{eq:regularized_psi2-2}
\end{align}
For every $\alpha\in(0,1)\cup(1,+\infty)$ and $\tp\in(0,+\infty]\cup\{\geo\}$,    
\begin{align}
&\red_{\alpha,\tp}^{\reg}(\omega_Q\|\omega_R)
=
\frac{1}{2\pi} \int_0^{2\pi}\red_{\alpha,\tp}(\what\omega_{\mtxsymfn{q}(x)}\|\what\omega_{\mtxsymfn{r}(x)})\,\dd x
\\
&= \frac{1}{2\pi} \int_0^{2\pi} \frac{1}{\alpha-1}\tr \Big[\log(I_d - \mtxsymfn{q}(x))^{\alpha} 
+ \log(I_d - \mtxsymfn{r}(x))^{1-\alpha} + \log\bz I_d + 
\req^{\op}_{\alpha,\tp}(\weightmtxfn{Q}(x)\| \weightmtxfn{R}(x))\jz\Big] \,\dd x.
\end{align}
\end{theorem}

\begin{theorem}
In the above setting, 
the regularized Umegaki relative entropy 
and the regularized Belavkin-Staszewski relative entropy 
of $\omega_Q$ and $\omega_R$ are given by 
\begin{align}
&D^{\reg}(\omega_Q \| \omega_R)\nn\\
&=
\frac{1}{2\pi} \int_0^{2\pi}D\bz\what\omega_{\mtxsymfn{q}(x)}\|\what\omega_{\mtxsymfn{r}(x)}\jz\,\dd x
=
\frac{1}{2\pi} \int_0^{2\pi}D^{\qf}\bz\mtxsymfn{q}(x)\|\mtxsymfn{r}(x)\jz\,\dd x
\label{eq:regularized_relentr1}\\
&=
\frac{1}{2\pi} \int_0^{2\pi} \tr 
\big[ \mtxsymfn{q}(x)\big[\log \mtxsymfn{q}(x)-\log\mtxsymfn{r}(x)\big]
+ 
(I_d - \mtxsymfn{q}(x))\big[\log(I_d - \mtxsymfn{q}(x))-
\log(I_d - \mtxsymfn{r}(x))\big]\big] \,\dd x, \label{eq:regularized_relentr2}\\
&D^{\reg}_{\geo}(\omega_Q \| \omega_R)\nn\\
&=
\frac{1}{2\pi} \int_0^{2\pi}D_{\geo}\bz\what\omega_{\mtxsymfn{q}(x)}\|\what\omega_{\mtxsymfn{r}(x)}\jz\,\dd x
=
\frac{1}{2\pi} \int_0^{2\pi}D_{\geo}^{\qf}\bz\mtxsymfn{q}(x)\|\mtxsymfn{r}(x)\jz\,\dd x
\label{eq:regularized_BS1}
\\
&=
\frac{1}{2\pi} \int_0^{2\pi}\Tr\big[\log\bz I_d-\mtxsymfn{q}(x)\jz
-
\log\bz I_d-\mtxsymfn{r}(x)\jz
+
\mtxsymfn{q}(x)
\log\big( (\weightmtxfn{Q}(x))^{1/2}(\weightmtxfn{R}(x))^{-1}(\weightmtxfn{Q}(x))^{1/2}\big)
\big]\,\dd x.
\label{eq:regularized_BS2}
\end{align}
\end{theorem}
\begin{proof}
According to \eqref{eq:qf Umegaki},
\begin{align*}
D^{\reg}(\omega_Q \| \omega_R)
&=
\sum_{k,l=0}^{1}(-1)^l\lim_{n\to+\infty}\frac{1}{n}\Tr
f^{(1)}\big(\big( A_{k,l}^{(1)}\big)_n\big) f^{(2)}\big(\big( A_{k,l}^{(2)}\big)_n\big)\\
&=
\sum_{k,l=0}^{1}(-1)^l\frac{1}{2\pi}\int_0^{2\pi}
\Tr f^{(1)}\big( \mtxsymfn{a}_{k,l}^{(1)}\big) f^{(2)}\big( \mtxsymfn{a}_{k,l}^{(2)}\big)\,\dd x,
\end{align*}
where $f^{(1)}(x)=x$, $f^{(2)}(x)=\log x$, $A^{(j)}_{k,l}=\ft^{-1}M_{\mtxsymfn{a}^{(j)}_{k,l}}\ft$,
\begin{align*}
\mtxsymfn{a}^{(1)}_{0,0}=
\mtxsymfn{a}^{(1)}_{0,1}=
\mtxsymfn{a}^{(2)}_{0,0}=\mtxsymfn{q},\ds
\mtxsymfn{a}^{(2)}_{0,1}=\mtxsymfn{r},\ds
\mtxsymfn{a}^{(1)}_{1,0}=
\mtxsymfn{a}^{(1)}_{1,1}=
\mtxsymfn{a}^{(2)}_{1,0}=\mathbf{1}-\mtxsymfn{q},\ds
\mtxsymfn{a}^{(2)}_{1,1}=\mathbf{1}-\mtxsymfn{r},
\end{align*}
and the second equality follows by Theorem \ref{thm:szego-block-poly}, since 
$\log$ is continuous on $\cup_{k,l\in\{0,1\}}\conv(\spec(A^{(2)}_{k,l}))$ by 
assumption \eqref{eq:strict bound2}.
This proves that $D^{\reg}(\omega_Q\|\omega_R)$ is equal to \eqref{eq:regularized_relentr2}, 
and the equality of \eqref{eq:regularized_relentr2} to the expressions in 
\eqref{eq:regularized_relentr1} follows by definition.

Similarly, by \eqref{eq:qf BS2},
\begin{align*}
D^{\reg}_{\geo}(\omega_Q \| \omega_R)
=&
\lim_{n\to+\infty}\frac{1}{n}\Tr\log(I - Q_n)-
\lim_{n\to+\infty}\frac{1}{n}\Tr \log (I- R_n)\\
&+\lim_{n\to+\infty}\frac{1}{n}\Tr \left[Q_n^{1/2}
\log\bz\W{Q_n}^{1/2}\W{R_n}\inv\W{Q_n}^{1/2}\jz Q_n^{1/2}\right]\\
=&
\frac{1}{2\pi} \int_0^{2\pi}\Tr\log\bz I_d-\mtxsymfn{q}(x)\jz\,\dd x
-
\frac{1}{2\pi} \int_0^{2\pi}\Tr\log\bz I_d-\mtxsymfn{r}(x)\jz\,\dd x\\
&+
\frac{1}{2\pi} \int_0^{2\pi}\Tr\mtxsymfn{q}(x)^{1/2}
\log\bz (\weightmtxfn{Q}(x))^{1/2}(\weightmtxfn{R}(x))^{-1}(\weightmtxfn{Q}(x))^{1/2}\jz
\mtxsymfn{q}(x)^{1/2}\,\dd x,
\end{align*}
where the second equality follows by straightforward applications of 
Theorem \ref{thm:szego-block-poly} for the first two terms, and of 
\eqref{eq:Szego for max} for the last term.
This proves that $D^{\reg}_{\geo}(\omega_Q\|\omega_R)$ is equal to \eqref{eq:regularized_BS2}, 
and the equality of \eqref{eq:regularized_BS2} to the expressions in 
\eqref{eq:regularized_BS1} follows by definition.
\end{proof}

\begin{lemma}\label{lem:psi-differentiability}
For any $\tp\in(0,+\infty]\cup\{\geo\}$, $\bR\ni\alpha\mapsto\psi^{\reg}_{\alpha,\tp}(\omega_Q\|\omega_R)$
is infinitely many times differentiable, and for every $k\in\bN$,
\begin{align}
\partial_{\alpha}^k\,\psi^{\reg}_{\alpha,\tp}(\omega_Q\|\omega_R)
=
\frac{1}{2\pi} \int_0^{2\pi}\partial_{\alpha}^k\,\psi_{\alpha,\tp}(\what\omega_{\mtxsymfn{q}(x)}\|\what\omega_{\mtxsymfn{r}(x)})\,\dd x.
\label{eq:regularized_psi_derivative}
\end{align}
Similarly, $(0,+\infty)\ni\alpha\mapsto\psi^{\reg}_{\alpha,\alpha}(\omega_Q\|\omega_R)$
is infinitely many times differentiable, and for every $k\in\bN$,
\begin{align}
\partial_{\alpha}^k\,\psi^{\reg}_{\alpha,\alpha}(\omega_Q\|\omega_R)
=
\frac{1}{2\pi} \int_0^{2\pi}\partial_{\alpha}^k\,\psi_{\alpha,\alpha}(\what\omega_{\mtxsymfn{q}(x)}\|\what\omega_{\mtxsymfn{r}(x)})\,\dd x.
\label{eq:regularized_psi_derivative sandwiched}
\end{align}
\end{lemma}
\begin{proof}
It is clear that $\psi_{\alpha,\tp}(\what\omega_{\mtxsymfn{q}(x)}\|\what\omega_{\mtxsymfn{r}(x)})$, or equivalently, the integrand in \eqref{eq:regularized_psi2}, is infinitely many times differentiable as a 
function of $\alpha$, and assumption \eqref{eq:strict bound1} guarantees that the derivatives of any order are bounded as a function on $\bT$, i.e., for any $k\in\bN$ and any compact interval 
$J\subset\bR$, 
$M_{k,J}:=\essup_{x\in\bT}\sup_{\alpha\in J}|\partial_{\alpha}^k\psi_{\alpha,\tp}(\what\omega_{\mtxsymfn{q}(x)}\|\what\omega_{\mtxsymfn{r}(x)})|<+\infty$. Hence, by the mean value theorem, the constant $M_{k,J}$ function is an integrable dominant to 
$(\partial_{\alpha}^{k-1}\psi_{\alpha',\tp}(\what\omega_{\mtxsymfn{q}(x)}\|\what\omega_{\mtxsymfn{r}(x)})
-\partial_{\alpha}^{k-1}\psi_{\alpha,\tp}(\what\omega_{\mtxsymfn{q}(x)}\|\what\omega_{\mtxsymfn{r}(x)}))/(\alpha'-\alpha)$ for any $\alpha',\alpha\in J$, $\alpha'\ne \alpha$. A standard application of the 
Lebesgue dominated convergence theorem then yields that 
the $k$-th derivative of $\psi^{\reg}_{\alpha,\tp}(\omega_Q\|\omega_R)$ exists, and is equal to 
the RHS of \eqref{eq:regularized_psi_derivative}.

This proves \eqref{eq:regularized_psi_derivative}.
The proof of \eqref{eq:regularized_psi_derivative sandwiched}
goes the same way, and hence we omit it.
\end{proof}

\begin{cor}\label{cor:regularized Renyi limit}
Let $\delta>0$ and $z:\,(1-\delta,1+\delta)\to(0,+\infty]$ be such that $\liminf_{\alpha\to 1}z(\alpha)>0$. Then 
\begin{align*}
\lim_{\alpha\to 1}D^{\reg}_{\alpha,z(\alpha)}(\omega_Q\|\omega_R)=D^{\reg}(\omega_Q\|\omega_R).
\end{align*}
\end{cor}
\begin{proof}
Let $\eta:=\liminf_{\alpha\to 1}z(\alpha)$. It is clear from the definition that for any 
$x\in\bT$, the function $(\alpha,z)\mapsto D_{\alpha,z}(\what\omega_{\mtxsymfn{q}(x)}\|\what\omega_{\mtxsymfn{r}(x)})$ is continuous on $((0,1)\cup(1,+\infty))\times[\eta/2,+\infty]$, and continuity at points of the form $(1,z)$, where $z\in[\eta/2,+\infty]$ follows from \eqref{eq:Umegaki def}. Moreover, assumption 
\eqref{eq:strict bound1} guarantees that 
\begin{align*}
\sup_{x\in\bT}\max_{\alpha\in[1-\delta/2,1+\delta/2]}\max_{z\in[\eta/2,+\infty]}
D_{\alpha,z}(\what\omega_{\mtxsymfn{q}(x)}\|\what\omega_{\mtxsymfn{r}(x)})<+\infty. 
\end{align*}
Hence, a standard application of the Lebesgue dominated convergence theorem yields that 
the limit and the integral below can be interchanged
\begin{align*}
\lim_{\alpha\to 1}D^{\reg}_{\alpha,z(\alpha)}(\omega_Q\|\omega_R)
&=
\lim_{\alpha\to 1}
\frac{1}{2\pi} \int_0^{2\pi}\,D_{\alpha,z(\alpha)}(\what\omega_{\mtxsymfn{q}(x)}\|\what\omega_{\mtxsymfn{r}(x)})\,\dd x\\
&=
\frac{1}{2\pi} \int_0^{2\pi}\,\lim_{\alpha\to 1}
D_{\alpha,z(\alpha)}(\what\omega_{\mtxsymfn{q}(x)}\|\what\omega_{\mtxsymfn{r}(x)})\,\dd x\\
&=
\frac{1}{2\pi} \int_0^{2\pi}D(\what\omega_{\mtxsymfn{q}(x)}\|\what\omega_{\mtxsymfn{r}(x)})\,\dd x
=D^{\reg}(\omega_Q\|\omega_R),
\end{align*}
and the rest follows by applying \eqref{eq:regularized_psi1-2} in the first,
and \eqref{eq:regularized_relentr1} in the last equality,
and \eqref{eq:Umegaki def} in the third equality.
\end{proof}

\begin{rem}
Corollary \ref{cor:regularized Renyi limit} shows that, in particular, both the regularized Petz-type and the regularized sandwiched R\'enyi $\alpha$-divergences converge to the regularized relative entropy in the 
$\alpha\to 1$ limit, i.e., 
\begin{align*}
\lim_{\alpha\to 1}D^{\reg}_{\alpha,1}(\omega_Q\|\omega_R)
=\lim_{\alpha\to 1}D^{\reg}_{\alpha,\alpha}(\omega_Q\|\omega_R)
=D^{\reg}(\omega_Q\|\omega_R).
\end{align*}
\end{rem}

\begin{cor}\label{cor:regularized max Renyi limit}
We have
\begin{align*}
\lim_{\alpha\to 1}D^{\reg}_{\alpha,\geo}(\omega_Q\|\omega_R)
=
D^{\reg}_{\geo}(\omega_Q\|\omega_R).
\end{align*}
\end{cor}
\begin{proof}
The proof is exactly analogous to that of Corollary \ref{cor:regularized Renyi limit},
using \eqref{eq:regularized_psi1-2} and \eqref{eq:regularized_BS1}; we omit the obvious details.
\end{proof}

\begin{rem}
Note that in the case of one mode per site, 
$\mtxsymfn{q}(x)=\hat q(x)\in[c,1-c]$,
$\mtxsymfn{r}(x)=\hat r(x)\in[c,1-c]$,
and 
\begin{align*}
\what\omega_{\mtxsymfn{q}(x)}=\begin{bmatrix}1-\hat q(x) & 0 \\ 0 & \hat q(x)\end{bmatrix},\ds\ds\ds
\what\omega_{\mtxsymfn{r}(x)}=\begin{bmatrix}1-\hat r(x) & 0 \\ 0 & \hat r(x)\end{bmatrix},
\end{align*}
are commuting qubit states. Hence, for any $\alpha\in\bR$, 
$\redi_{\alpha,\gamma}\bz\what\omega_{\mtxsymfn{q}(x)}\|\what\omega_{\mtxsymfn{r}(x)}\jz$ is independent of the value of $\gamma$, and therefore so is $\redi_{\alpha,\tp}^{\reg}(\omega_Q\|\omega_R)$, according to Theorem 
\ref{thm:psi-existence-2}. This is a manifestation of the asymptotic commutativity of 
$\omega_Q$ and $\omega_R$ in the single-mode case. Note that finite-size commutativity does not hold in the sense that $\what\omega_{Q_n}$ and $\what\omega_{R_n}$ do not commute in general, and hence 
the value of $\redi_{\alpha,\gamma}\bz\what\omega_{Q_n}\|\what\omega_{R_n}\jz$ depends on $\gamma$; see, e.g., 
\cite[Theorem 2.1]{Hiai-ALT}.
\end{rem}

\subsubsection{Regularized measured R\'enyi divergences}

For any notion of quantum R\'enyi $\alpha$-divergence $D_{\alpha,q}$, one may define its regularized version
in a more general setting along arbitrary sequences of pairs of PSD operators. That is, 
assume that for every $n\in\bN$, $\rho_n,\sigma_n$ are non-zero PSD operators on some finite-dimensional 
Hilbert space $\hil_n$. We will use the notations $\sequence{\rho}:=(\rho_n)_{n\in\bN}$, 
$\sequence{\sigma}:=(\sigma_n)_{n\in\bN}$.
Three notions of regularized $q$-R\'enyi $\alpha$-divergence are then defined as
\begin{align*}
\ul \redi^{\reg}_{\alpha,q}(\sequence{\rho}\|\sequence{\sigma})
&:=
\liminf_{n\to+\infty}\frac{1}{n}\redi_{\alpha,q}(\rho_n\|\sigma_n),\\
\oll \redi^{\reg}_{\alpha,q}(\sequence{\rho}\|\sequence{\sigma})
&:=
\limsup_{n\to+\infty}\frac{1}{n}\redi_{\alpha,q}(\rho_n\|\sigma_n),\\
\redi^{\reg}_{\alpha,q}(\sequence{\rho}\|\sequence{\sigma})
&:=
\lim_{n\to+\infty}\frac{1}{n}\redi_{\alpha,q}(\rho_n\|\sigma_n),
\end{align*}
where the last quantity is only defined when the limit exists, or equivalently,
$\ul \redi^{\reg}_{\alpha,q}(\sequence{\rho}\|\sequence{\sigma})
=
\oll \redi^{\reg}_{\alpha,q}(\sequence{\rho}\|\sequence{\sigma})$.
The regularized R\'enyi divergences 
$\ul \red^{\reg}_{\alpha,q}(\sequence{\rho}\|\sequence{\sigma}),
\oll \red^{\reg}_{\alpha,q}(\sequence{\rho}\|\sequence{\sigma})$, and 
$\red^{\reg}_{\alpha,q}(\sequence{\rho}\|\sequence{\sigma})$
are defined analogously.

We will consider such regularized quantities for two further notions of quantum R\'enyi divergences beyond the ones considered in Section \ref{sec:regularized az}: the measured R\'enyi divergences in this section, and 
the integral R\'enyi divergences in Section \ref{sec:regularized integral Renyi}.

For any pair of non-zero PSD operators $\rho,\sigma\in\B(\hil)\pne$, their
\ki{measured R\'enyi $\alpha$-divergence} for some $\alpha\in\bR$ is defined as
\begin{align*}
\redi_{\alpha,\meas}(\rho\|\sigma):=\max\set{\redi_{\alpha}\bz\M(\rho)\|\M(\sigma)\jz|M\in\povm(\hil,[d]),\,d\in\bN}.
\end{align*}
Note that here $\M(\rho)$ and $\M(\sigma)$ commute, and hence 
for any fixed $\alpha$, all previously considered notions of R\'enyi
$\alpha$-quantities $\req_{\alpha,\gamma}$, $\gamma\in(0,+\infty]\cup\{\geo\}$ coincide on them, 
and are equal to the classical R\'enyi $\alpha$-quantity \cite{Renyi} of 
$(\Tr M_i\rho)_{i\in[d]}$ and 
$(\Tr M_i\sigma)_{i\in[d]}$.
Hence, we simply write $D_{\alpha}\bz\M(\rho)\|\M(\sigma)\jz$ for this unique 
R\'enyi $\alpha$-divergence.
It is well known and easy to verify that the measured R\'enyi divergences are quantum R\'enyi divergences in the sense explained at the beginning of Section \ref{sec:finitedim}.
Analogously, we define for every $\alpha\in(0,+\infty)$,
\begin{align*}
\red_{\alpha,\meas}(\rho\|\sigma)
&:=
\max\set{\red_{\alpha}\bz\M(\rho)\|\M(\sigma)\jz|M\in\povm(\hil,[d]),\,d\in\bN}.
\end{align*}
According to \eqref{eq:Renyi scaling2}, for any $\alpha\in(0,1)\cup(1,+\infty)$, we have
\begin{align*}
\red_{\alpha,\meas}(\rho\|\sigma)
&=
\frac{\oll\alpha-1}{\alpha-1}\redi_{\alpha,\meas}(\rho\|\sigma)+\log\Tr\rho-\log\Tr\sigma.
\end{align*}

For any $\alpha\in[1/2,+\infty)$, let 
\begin{align*}
\kappa(\alpha):=\begin{cases}
1,&\alpha\in[1/2,2],\\
\frac{\alpha}{\alpha-1},&\alpha>2.
\end{cases}
\end{align*}
According to \cite[Lemma 3]{HT14} and the proof of \cite[Theorem 3.7]{MO},
for any $\rho,\omega\in\B(\hil)\pne$,
\begin{align}
\redi_{\alpha,\alpha}(\rho\|\omega)-\kappa(\alpha)\log|\spec(\omega)|
&\le
\redi_{\alpha,\meas}(\rho\|\omega)
\le
\redi_{\alpha,\alpha}(\rho\|\omega),\label{eq:pinching bound}\\
\red_{\alpha,\alpha}(\rho\|\omega)-\kappa(\alpha)\log|\spec(\omega)|
&\le
\red_{\alpha,\meas}(\rho\|\omega)
\le
\red_{\alpha,\alpha}(\rho\|\omega),
\ds\ds\ds\alpha\in[1/2,+\infty),
\label{eq:pinching bound2}
\end{align}
where $|\spec(\omega)|$ denotes the number of different eigenvalues of $\omega$.
From this it follows immediately that for a sequence of i.i.d. states $\rho_n=\rho^{\otimes n}$, 
$\sigma_n=\sigma^{\otimes n}$, $n\in\bN$, the regularized measured R\'enyi $\alpha$-divergence along this sequence 
is equal to the sandwiched R\'enyi $\alpha$-divergence of $\rho$ and $\sigma$ for any $\alpha\in[1/2,+\infty)$; see \cite{HT14,MO}. This is due to the fact that $|\spec(\sigma^{\otimes n})|$ grows only polynomially in $n$. 
This latter property, however, need not hold anymore for $|\spec(\what\omega_{Q_n})|$ for a quasi-free state
$\omega_Q$ of a fermion chain, which may grow exponentially in $n$. 

An alternative bound to circumvent this problem was given in \cite[Lemma 10, Example 11]{MO-correlated}, based on a technique of grouping the eigenvalues, introduced in \cite[Theorem 14]{TomamichelHayashi2013}.
Lemma \ref{lemma:measured vs sandwiched} below is a variant and slight improvement of \cite[Lemma 10, Example 11]{MO-correlated}. To state it, we 
follow \cite{TomamichelHayashi2013} and introduce for any non-zero PSD operator 
$\omega\in\B(\hil)\pne$ the quantity
\begin{align*}
\Theta(\omega):=\min\left\{|\spec(\omega)|,2+\ceil{\log\lambda_{\max}(\omega)-\log\lambda_{\min}^*(\omega)}\right\},
\end{align*}
where $\lambda_{\max}(\omega)=\norm{\omega}$ is the largest eigenvalue of $\omega$, and 
\begin{align}\label{eq:lambda min}
\lambda_{\min}^*(\omega):=\max\set{\lambda>0|\lambda\omega^0\le\omega}
\end{align}
is the smallest non-zero eigenvalue of $\omega$. 

\begin{lemma}\label{lemma:measured vs sandwiched}
For any non-zero PSD operators $\rho,\sigma\in\B(\hil)\pne$ on some finite-dimensional
Hilbert space $\hil$,
\begin{align}
\redi_{\alpha,\alpha}(\rho\|\sigma)-\kappa(\alpha)\log\Theta(\sigma)-1
&\le
\redi_{\alpha,\meas}(\rho\|\sigma)
\le 
\redi_{\alpha,\alpha}(\rho\|\sigma),
&\alpha\in[1/2,+\infty),
\label{eq:measured Renyi bounds1}\\
\redi_{\alpha,1-\alpha}(\rho\|\sigma)-\kappa(1-\alpha)\log\Theta(\rho)-1
&\le
\redi_{\alpha,\meas}(\rho\|\sigma)
\le
\redi_{\alpha,1-\alpha}(\rho\|\sigma),
&\alpha\in(-\infty,1/2].\label{eq:measured Renyi bounds2}
\end{align}
\end{lemma}
\begin{proof}
Assume first that $\alpha\in[1/2,+\infty)$. 
Then, the second inequality in \eqref{eq:measured Renyi bounds1} follows from the monotonicity 
of $\redi_{\alpha,\alpha}$ under CPTP maps; see, e.g., \cite{FL}. 
Applying \eqref{eq:pinching bound} with $\omega=\sigma$ yields
\begin{align}\label{eq:lemma:measured vs sandwiched proof1}
\redi_{\alpha,\alpha}(\rho\|\sigma)-\kappa(\alpha)\log|\spec(\sigma)|
&\le
\redi_{\alpha,\meas}(\rho\|\sigma).
\end{align}
Hence, the proof of \eqref{eq:measured Renyi bounds1} will be complete if we show that 
\begin{align}\label{eq:lemma:measured vs sandwiched proof2}
\redi_{\alpha,\alpha}(\rho\|\sigma)-\kappa(\alpha)\log(2+m)-1
&\le
\redi_{\alpha,\meas}(\rho\|\sigma),
\end{align}
where $m:=\ceil{\log q}$, $q:=\lambda_{\max}(\sigma)/\lambda_{\min}^*(\sigma)$.
For this, we follow the proof idea of 
\cite[Theorem 14]{TomamichelHayashi2013}.

First, note that if $m=0$ then \eqref{eq:lemma:measured vs sandwiched proof2} follows from 
\eqref{eq:lemma:measured vs sandwiched proof1}. Hence, for the rest we assume that $m>0$, or equivalently, that 
$\lambda_{\max}(\sigma)>\lambda_{\min}^*(\sigma)$.
Define
\begin{align*}
\what\sigma:=\sum_{\lambda\in\spec(\sigma)\setminus\{0\}}\lambda_{\min}^*(\sigma)q^{\frac{k(\lambda)}{m}}P_{\lambda}^{\sigma},
\ds\ds\ds\text{where}\ds\ds\ds
k(\lambda):=
\ceil{m\frac{\log\lambda-\log\lambda_{\min}^*(\sigma)}{\log\lambda_{\max}(\sigma)-\log\lambda_{\min}^*(\sigma)}},
\end{align*}
and $P_{\lambda}^{\sigma}$ is the spectral projection of $\sigma$ corresponding to the eigenvalue $\lambda$.
Then 
\begin{align*}
\sigma\le\what\sigma\le q^{\frac{1}{m}}\sigma,\ds\ds\ds\text{and}\ds\ds\ds
|\spec(\what\sigma)|\le m+2.
\end{align*}
Hence,
\begin{align*}
\red_{\alpha,\alpha}(\rho\|\sigma)
&\le
\red_{\alpha,\alpha}(\rho\|q^{-1/m}\what\sigma)
=
\red_{\alpha,\alpha}(\rho\|\what\sigma)+\underbrace{\frac{\log q}{m}}_{\le 1}\\
&\le
\red_{\alpha,\meas}(\rho\|\what\sigma)+\kappa(\alpha)\log\underbrace{|\spec(\what\sigma)|}_{\le m+2}+1
\le
\red_{\alpha,\meas}(\rho\|\sigma)+\kappa(\alpha)\log(m+2)+1,
\end{align*}
where the first and the last inequalities follow as $\red_{\alpha,\alpha}$ and 
$\red_{\alpha,\meas}$ are both
monotone non-increasing in their second argument w.r.t. the PSD order, due to the 
operator monotonicity properties of $\id_{[0,+\infty)}^{\frac{1-\alpha}{\alpha}}$ for the given $\alpha$ values,
the equality is straightforward by definition, and the second inequality follows by applying 
\eqref{eq:pinching bound2} with $\omega=\what\sigma$.
Taking into account \eqref{eq:Renyi scaling2}, we obtain \eqref{eq:lemma:measured vs sandwiched proof2}.

Next, consider $\alpha\in(-\infty,1/2]$. It is 
obvious by definition that the symmetry relation \eqref{eq:Renyi symmetry} extends to the measured 
R\'enyi divergence as well.
Hence,
\begin{align*}
\redi_{\alpha,\meas}(\rho\|\sigma)&=\redi_{1-\alpha,\meas}(\sigma\|\rho)
\ge
\redi_{1-\alpha,1-\alpha}(\sigma\|\rho)-\kappa(1-\alpha)\log\Theta(\rho)-1,
\end{align*}
where the inequality is due to the first inequality in \eqref{eq:measured Renyi bounds1}.
Likewise,
\begin{align*}
\redi_{\alpha,\meas}(\rho\|\sigma)&=
\redi_{1-\alpha,\meas}(\sigma\|\rho)
\le
\redi_{1-\alpha,1-\alpha}(\sigma\|\rho)=
\redi_{\alpha,1-\alpha}(\rho\|\sigma),
\end{align*}
where the inequality is due to the second inequality in \eqref{eq:measured Renyi bounds1}.
\end{proof}

\begin{rem}
For $\alpha\ge 1$, the correction factor $\Theta(\sigma)$ in \eqref{eq:measured Renyi bounds1} can be improved to 
\begin{align*}
\Theta^*(\sigma):=\min\left\{|\spec(\sigma)\setminus\{0\}|,1+\ceil{\log\lambda_{\max}(\sigma)-
\log\lambda_{\min}^*(\sigma)}\right\}.
\end{align*}
This is because it only plays a role when $\rho^0\le\sigma^0$, since otherwise 
$\redi_{\alpha,\alpha}(\rho\|\sigma)=+\infty=\redi_{\alpha,\meas}(\rho\|\sigma)$. On the other hand,
\eqref{eq:pinching bound} follows from the pinching inequality \cite{H:pinching}, according to which
\begin{align*}
\rho\le|\spec(\omega)|\P_{\omega}(\rho), 
\end{align*}
where $\P_{\omega}(\rho):=\sum_{\lambda\in\spec(\omega)}P_{\lambda}^{\omega}\rho P_{\lambda}^{\omega}$. 
However, the prefactor here can be improved to $|\spec(\omega)\setminus\{0\}|=:r$ when 
$\rho^0\le \omega^0$. 
Indeed, the map $\B(\hil)\ni X\mapsto X^*\rho X$ is easily seen to be operator convex, whence
\begin{align*}
\rho=r^2\bz\sum_{\lambda\in\spec(\omega)\setminus\{0\}}\frac{1}{r}P_{\lambda}^{\omega}\jz
\rho\bz\sum_{\lambda\in\spec(\omega)\setminus\{0\}}\frac{1}{r}P_{\lambda}^{\omega}\jz
\le r\sum_{\lambda\in\spec(\omega)\setminus\{0\}}P_{\lambda}^{\omega}\rho P_{\lambda}^{\omega}.
\end{align*}
\end{rem}

\begin{cor}\label{cor:measured vs sandwiched}
For any two sequences of non-zero PSD operators $\rho_n,\sigma_n\in\B(\hil_n)\pne$, $n\in\bN$, 
\begin{align*}
\ul \redi^{\reg}_{\alpha,\meas}(\sequence{\rho}\|\sequence{\sigma})
&\le
\begin{cases}
\ul \redi^{\reg}_{\alpha,\alpha}(\sequence{\rho}\|\sequence{\sigma}),&\alpha\in[1/2,+\infty),\\
\ul \redi^{\reg}_{\alpha,1-\alpha}(\sequence{\rho}\|\sequence{\sigma}),&\alpha\in(-\infty,1/2],
\end{cases}\\
\oll \redi^{\reg}_{\alpha,\meas}(\sequence{\rho}\|\sequence{\sigma})
&\le
\begin{cases}
\oll \redi^{\reg}_{\alpha,\alpha}(\sequence{\rho}\|\sequence{\sigma}),&\alpha\in[1/2,+\infty),\\
\oll \redi^{\reg}_{\alpha,1-\alpha}(\sequence{\rho}\|\sequence{\sigma}),&\alpha\in(-\infty,1/2],
\end{cases}.
\end{align*}
Moreover, if $\lim_{n\to+\infty}\frac{1}{n}\log\Theta(\sigma_n)=0$ then, 
for any $\alpha\in[1/2,+\infty)$, 
$\redi_{\alpha,\meas}^{\reg}(\sequence{\rho}\|\sequence{\sigma})$ exists if and only if 
$\redi_{\alpha,\alpha}^{\reg}(\sequence{\rho}\|\sequence{\sigma})$ exists, in which case the two are equal, i.e.,
\begin{align}\label{eq:regmeas equals sandwiched1}
\redi^{\reg}_{\alpha,\meas}(\sequence{\rho}\|\sequence{\sigma})=
\redi^{\reg}_{\alpha,\alpha}(\sequence{\rho}\|\sequence{\sigma}).
\end{align}

Likewise, if $\lim_{n\to+\infty}\frac{1}{n}\log\Theta(\rho_n)=0$ then 
for any $\alpha\in(-\infty,1/2]$,
$\redi_{\alpha,\meas}^{\reg}(\sequence{\rho}\|\sequence{\sigma})$ exists if and only if 
$\redi_{\alpha,1-\alpha}^{\reg}(\sequence{\rho}\|\sequence{\sigma})$ exists, in which case the two are equal, i.e., 
\begin{align}\label{eq:regmeas equals sandwiched2}
\redi^{\reg}_{\alpha,\meas}(\sequence{\rho}\|\sequence{\sigma})=
\redi^{\reg}_{\alpha,1-\alpha}(\sequence{\rho}\|\sequence{\sigma}).
\end{align}
\end{cor}
\begin{proof}
Immediate from Lemma \ref{lemma:measured vs sandwiched}.
\end{proof}

\begin{thm}\label{thm:regmeas Renyi}
Assume that $Q,R\in\B(\ell^2_d(\bZ))$ satisfy \eqref{eq:strict bound2}, i.e.,  $cI\le Q,R\le (1-c)I$. Then 
\begin{align}\label{eq:theta zero}
\lim_{n\to+\infty}\frac{1}{n}\log\Theta(\what\omega_{Q_n})=0=\lim_{n\to+\infty}\frac{1}{n}\log\Theta(\what\omega_{R_n})=0,
\end{align}
and
\begin{align}
\redi^{\reg}_{\alpha,\meas}(\omega_Q\|\omega_R)=
\begin{cases}
\redi^{\reg}_{\alpha,\alpha}(\omega_Q\|\omega_R)
=
\frac{1}{2\pi} \int_0^{2\pi}D_{\alpha,\alpha}(\what\omega_{\mtxsymfn{q}(x)}\|\what\omega_{\mtxsymfn{r}(x)})\,\dd x
,&\alpha\in[1/2,+\infty),\\
\redi^{\reg}_{\alpha,1-\alpha}(\omega_Q\|\omega_R)
=
\frac{1}{2\pi} \int_0^{2\pi}D_{\alpha,1-\alpha}(\what\omega_{\mtxsymfn{q}(x)}\|\what\omega_{\mtxsymfn{r}(x)})\,\dd x
,&\alpha\in(-\infty,1/2].
\end{cases}
\label{eq:regmeasured explicit}
\end{align}
\end{thm}
\begin{proof}
By assumption \eqref{eq:strict bound2}, 
$cI_{nd}\le R_n\le (1-c)I_{nd}$. According to the explanation between \eqref{quasifree density} and 
\eqref{eq:quasifree density},
these bounds imply that every eigenvalue of $\what\omega_{R_n}$ is a product of $nd$ factors, each lower bounded by $c$, whence $\what\omega_{R_n}\ge c^{nd}I_{nd}$.
Thus,
\begin{align*}
\Theta(\what\omega_{R_n})\le\bz 3+\log\lambda_{\max}(\what\omega_{R_n})-\log\lambda_{\min}(\what\omega_{R_n})\jz
\le\bz 4-dn\log c\jz,
\end{align*}
and
\begin{align*}
0\le\liminf_{n\to+\infty}\frac{1}{n}\log\Theta(\what\omega_{R_n})\le\limsup_{n\to+\infty}\frac{1}{n}\log\Theta(\what\omega_{R_n})\le 0,
\end{align*}
and an exactly analogous argument yields that 
$\lim_{n\to+\infty}\frac{1}{n}\log\Theta(\what\omega_{Q_n})=0$,
proving \eqref{eq:theta zero}.
The assertion in \eqref{eq:regmeasured explicit} then follows immediately from 
Corollary \ref{cor:measured vs sandwiched} and Theorem \ref{thm:psi-existence-2}.
\end{proof}

\subsubsection{Regularized hockey-stick R\'enyi divergences}
\label{sec:regularized integral Renyi}

In this section we consider the recently introduced \ki{integral R\'enyi divergences}, or 
\ki{hockey-stick R\'enyi divergences} \cite{Frenkel_integral,Hirche_Tomamichel_integral}.
According to \cite{HMT_hs}, for any $\rho,\sigma\in\B(\hil)\pne$ and $\alpha\in(0,1)\cup(1,+\infty)$,
the hockey-stick R\'enyi $\alpha$-quantity of $\rho$ and $\sigma$ can be written as
\begin{align}\label{eq:hsq def}
\hsq{\alpha}{}(\rho\|\sigma)
&:=
\Tr\sigma+\alpha(\alpha-1)\left[\int_0^1\Tr(\rho-t\sigma)_-\,t^{\alpha-2}\,dt
+\int_1^{+\infty}\Tr(\rho-t\sigma)_+\,t^{\alpha-2}\,dt\right].
\end{align}
Here, $(\rho-t\sigma)_-$ is the negative part and $(\rho-t\sigma)_+$ is the positive part of the self-adjoint
operator $\rho-t\sigma$. 
These satisfy $\hsq{\alpha}{}(\rho\|\sigma)=\hsq{1-\alpha}{}(\sigma\|\rho)$ for $\alpha\in(0,1)$, and hence we extend the definition as
\begin{align}\label{eq:hsq def2}
\hsq{\alpha}{}(\rho\|\sigma):=\hsq{1-\alpha}{}(\sigma\|\rho),\ds\ds\ds\alpha\in(-\infty,0),
\end{align}
and define the hockey-stick R\'enyi $\alpha$-divergence of $\rho$ and $\sigma$ as
\begin{align*}
\hsd{\alpha}{}(\rho\|\sigma):=\frac{1}{\oll\alpha-1}\log\hsq{\alpha}{}\bz\frac{\rho}{\Tr\rho}\Big\|
\frac{\sigma}{\Tr\sigma}\jz,\ds\ds\ds\alpha\in\bR\setminus\{0,1\}.
\end{align*}
These are quantum R\'enyi $\alpha$-divergences in the sense that for commuting operators 
$\rho,\sigma$, $\hsd{\alpha}{}(\rho\|\sigma)$ is the classical R\'enyi $\alpha$-divergence of the diagonal elements of the matrices of $\rho$ and $\sigma$ in any orthonormal basis diagonalizing both; 
see Section \ref{sec:finitedim}. According to \cite[Proposition 18]{Frenkel_integral},
\begin{align*}
\hsd{1}{}(\rho\|\sigma)&:=\lim_{\alpha\to 1}\hsd{\alpha}{}(\rho\|\sigma)=D\bz\frac{\rho}{\Tr\rho}\Big\|
\frac{\sigma}{\Tr\sigma}\jz,\\
\hsd{0}{}(\rho\|\sigma)&:=\lim_{\alpha\to 0}\hsd{\alpha}{}(\rho\|\sigma)
=
D\bz\frac{\sigma}{\Tr\sigma}\Big\|\frac{\rho}{\Tr\rho}\jz.
\end{align*}

The following bounds have been shown in \cite{BeigiHircheTomamichel2025,Hirche_Tomamichel_integral,Liu_etal_layercake}. 
More precisely, the first inequality in \eqref{eq:hs bounds1} was shown in 
\cite[Proposition 3.9]{BeigiHircheTomamichel2025}, the second inequality in \eqref{eq:hs bounds1}
is Eq.~(3.55) in \cite{Hirche_Tomamichel_integral}, 
the second inequality in \eqref{eq:hs bounds2} is from \cite[Proposition 4.5]{Liu_etal_layercake},
and the first inequality in \eqref{eq:hs bounds2} is due to the fact that 
the hockey-stick R\'enyi divergences are manifestly monotone under positive trace-preserving maps,
and the measured R\'enyi divergences are the smallest R\'enyi divergences that are monotone under 
 positive trace-preserving maps, which can be easily verified from their definition.
Moreover, the inequalities were proved in \cite{BeigiHircheTomamichel2025,Hirche_Tomamichel_integral,Liu_etal_layercake}
for density operators; the forms below follow in a straightforward way
according to the definitions above.
\begin{lemma}
For any $\rho,\sigma\in\S(\hil)$,
\begin{align}
\redi_{\alpha,1}(\rho\|\sigma)&\le\hsd{\alpha}{}(\rho\|\sigma)\le \redi_{\alpha,1}(\rho\|\sigma)+\frac{\log 2}{1-\oll\alpha},& &
\alpha\in(0,1),\label{eq:hs bounds1}\\
\redi_{\alpha,\meas}(\rho\|\sigma)&\le\hsd{\alpha}{}(\rho\|\sigma)\le \redi_{\alpha,\alpha}(\rho\|\sigma),
& &\alpha\in[1,+\infty).
\label{eq:hs bounds2}
\end{align}
\end{lemma}

\begin{thm}
Assume that $Q,R\in\B(\ell^2_d(\bZ))$ satisfy \eqref{eq:strict bound2}, i.e.,  $cI\le Q,R\le (1-c)I$. Then 
\begin{align*}
\mathcal{D}^{\reg}_{\alpha,\hs}(\omega_Q\|\omega_R)=
\begin{cases}
\redi^{\reg}_{\alpha,1}(\omega_Q\|\omega_R)
=
\frac{1}{2\pi} \int_0^{2\pi}\redi_{\alpha,1}(\what\omega_{\mtxsymfn{q}(x)}\|\what\omega_{\mtxsymfn{r}(x)})\,\dd x
,&\alpha\in(0,1),\\
\redi^{\reg}_{\alpha,\alpha}(\omega_Q\|\omega_R)
=
\frac{1}{2\pi} \int_0^{2\pi}\redi_{\alpha,\alpha}(\what\omega_{\mtxsymfn{q}(x)}\|\what\omega_{\mtxsymfn{r}(x)})\,\dd x
,&\alpha\in[1,+\infty),\\
\redi^{\reg}_{\alpha,1-\alpha}(\omega_Q\|\omega_R)
=
\frac{1}{2\pi} \int_0^{2\pi}\redi_{\alpha,1-\alpha}(\what\omega_{\mtxsymfn{q}(x)}\|\what\omega_{\mtxsymfn{r}(x)})\,\dd x
,&\alpha\in(-\infty,0].
\end{cases}
\end{align*}
\end{thm}
\begin{proof}
In the case $\alpha\in(0,1)$, the inequalities in \eqref{eq:hs bounds1} yield immediately that 
$\mathcal{D}_{\alpha,\hs}^{\reg}(\omega_Q\|\omega_R)$ exists and is equal to 
$\redi_{\alpha,1}^{\reg}(\omega_Q\|\omega_R)$, 
whence the assertion follows from Theorem \ref{thm:psi-existence-2}.

In the case $\alpha\in[1,+\infty)$,
\begin{align*}
\redi_{\alpha,\alpha}^{\reg}(\omega_Q\|\omega_R)=
\redi_{\alpha,\meas}^{\reg}(\omega_Q\|\omega_R)
&\le
\underline{\mathcal{D}}_{\alpha,\hs}(\omega_Q\|\omega_R)
\le
\overline{\mathcal{D}}_{\alpha,\hs}(\omega_Q\|\omega_R)
\le \redi_{\alpha,\alpha}^{\reg}(\omega_Q\|\omega_R),
\end{align*}
where the equality is due to Theorem \ref{thm:regmeas Renyi}, the second and the fourth inequalities
follow from the inequalities in \eqref{eq:hs bounds2},
and the third inequality is trivial. Hence, the assertion in this case follows from 
Theorem \ref{thm:psi-existence-2}.

Finally, the case $\alpha\in(-\infty,0]$ follows from the above and the symmetry relations \eqref{eq:Renyi symmetry} and 
$\hsd{\alpha}{}(\rho\|\sigma)=\hsd{1-\alpha}{}(\sigma\|\rho)$.
\end{proof}

\subsection{Error exponents}
\label{sec:errexp}

Let us now turn to the problem of asymptotic discrimination of two quasi-free states $\omega_Q$ and $\omega_R$
on $\car{\ell^2_d(\bZ)}$ as explained in the Introduction. Recall that for a test 
$T_n\in\B\bz\fock{\ell^2_d(\n)}\jz_{[0,1]}$, the corresponding type I and type II error
probabilities are defined as
\begin{align*}
\err_0(\omega_{Q_n}|T_n):=\Tr\what\omega_{Q_n}(I-T_n),\ds\ds\text{(type I)},\ds\ds\ds\ds\ds
\err_1(\omega_{R_n}|T_n):=\Tr\what\omega_{R_n} T_n\ds\ds\text{(type II)},
\end{align*} 
where $Q_n:=(P_n\otimes I_d)^*Q(P_n\otimes I_d)$, 
$R_n:=(P_n\otimes I_d)^*R(P_n\otimes I_d)$, as in \eqref{eq:cutoff def}.
We define the direct exponents corresponding to a fixed type II error exponent $r$ as 
\begin{align}
\dls_r(\omega_Q\|\omega_R)
&:=
\sup\left\{\limsup_{n\to+\infty}-\frac{1}{n}\log\err_0(\omega_{Q_n}|T_n)\,\Big|\,
\liminf_{n\to+\infty}-\frac{1}{n}\log\err_1(\omega_{R_n}|T_n)>r
\right\},\label{eq:direct exp def1}\\
\dli_r(\omega_Q\|\omega_R)
&:=
\sup\left\{\liminf_{n\to+\infty}-\frac{1}{n}\log\err_0(\omega_{Q_n}|T_n)\,\Big|\,
\liminf_{n\to+\infty}-\frac{1}{n}\log\err_1(\omega_{R_n}|T_n)>r
\right\},\label{eq:direct exp def2}\\
\dl_r(\omega_Q\|\omega_R)
&:=
\sup\left\{\lim_{n\to+\infty}-\frac{1}{n}\log\err_0(\omega_{Q_n}|T_n)\,\Big|\,
\liminf_{n\to+\infty}-\frac{1}{n}\log\err_1(\omega_{R_n}|T_n)>r
\right\},\label{eq:direct exp def3}
\end{align}
where in the last definition we only optimize over test sequences for which the indicated limit exists. Obviously,
\begin{align*}
\dl_r(\omega_Q\|\omega_R)\le\dli_r(\omega_Q\|\omega_R)\le\dls_r(\omega_Q\|\omega_R).
\end{align*}
Similarly, the strong converse exponents corresponding to a fixed type II error exponent $r$ are defined as 
\begin{align}
\scs_r(\omega_Q\|\omega_R)
&:=
\inf\left\{\limsup_{n\to+\infty}-\frac{1}{n}\log(1-\err_0(\omega_{Q_n}|T_n))\,\Big|\,
\liminf_{n\to+\infty}-\frac{1}{n}\log\err_1(\omega_{R_n}|T_n)>r
\right\},\label{eq:sc exp def1}\\
\scli_r(\omega_Q\|\omega_R)
&:=
\inf\left\{\liminf_{n\to+\infty}-\frac{1}{n}\log(1-\err_0(\omega_{Q_n}|T_n))\,\Big|\,
\liminf_{n\to+\infty}-\frac{1}{n}\log\err_1(\omega_{R_n}|T_n)>r
\right\},\label{eq:sc exp def2}\\
\sc_r(\omega_Q\|\omega_R)
&:=
\inf\left\{\lim_{n\to+\infty}-\frac{1}{n}\log(1-\err_0(\omega_{Q_n}|T_n))\,\Big|\,
\liminf_{n\to+\infty}-\frac{1}{n}\log\err_1(\omega_{R_n}|T_n)>r
\right\},\label{eq:sc exp def3}
\end{align}
and
\begin{align*}
\scli_r(\omega_Q\|\omega_R)\le\scs_r(\omega_Q\|\omega_R)\le \sc_r(\omega_Q\|\omega_R)
\end{align*}
holds trivially.

\begin{theorem}
Let $\omega_Q$ and $\omega_R$ be quasi-free states satisfying \eqref{eq:strict bound2}, i.e., 
$cI\le Q,R\le (1-c)I$ for some $c\in(0,1/2)$. Then 
\begin{align}\label{eqn:direct-exp}
\dl_r(\omega_Q\|\omega_R)=\dli_r(\omega_Q\|\omega_R)=\dls_r(\omega_Q\|\omega_R)=
H_r (\omega_Q\|\omega_R):=\sup_{\alpha\in(0,1)}\frac{\alpha-1}{\alpha}\left[r-\red_{\alpha,1}^{\reg}(\omega_Q\|\omega_R)\right],
\end{align}
and
\begin{align}\label{eqn:strong-converse-exp}
\sc_r(\omega_Q\|\omega_R)=\scli_r(\omega_Q\|\omega_R)=\scs_r(\omega_Q\|\omega_R)=
H_r^* (\omega_Q\|\omega_R):=\sup_{\alpha>1}\frac{\alpha-1}{\alpha}\left[r-\red_{\alpha,\alpha}^{\reg}(\omega_Q\|\omega_R)\right],
\end{align}
for every $r\in(0,+\infty)$.
\end{theorem}
\begin{proof}
By Theorem \ref{thm:psi-existence} $\psi^{\reg}_{\alpha,1}(\omega_Q\|\omega_R)$ 
exists for every $\alpha\in\bR$, and by Lemma \ref{lem:psi-differentiability},
it is a differentiable function of $\alpha$ on $\bR$.
Hence, we may apply \cite[Theorem 4.8]{HMO2} to obtain \eqref{eqn:direct-exp}. 
Here, we note that \cite{HMO2} uses a different convention to define the error exponents, where the roles of the two types of errors are interchanged, and therefore \cite[Theorem 4.8]{HMO2} has to be applied accordingly.

Likewise, by Theorem \ref{thm:psi-existence}, $\psi^{\reg}_{\alpha,\alpha}(\omega_Q\|\omega_R)$ 
exists for every $\alpha\in(0,+\infty)$, and by Lemma \ref{lem:psi-differentiability}, it is a
differentiable function of $\alpha$ on $(0,+\infty)$. Moreover, by \eqref{eq:theta zero}, 
$\lim_{n\to+\infty}\frac{1}{n}\log\Theta(\what\omega_{R_n})=0$, and hence \cite[Corollary 22]{MO-correlated}
yields \eqref{eqn:strong-converse-exp}.
\end{proof}

\section{Super-exponential state discrimination}
\label{sec:superexp}

We say that two sequences of states $\rho_n,\sigma_n\in\S(\hil_n)$, $n\in\bN$, can be 
super-exponentially discriminated, if there exists a test sequence $T_n\in\B(\hil_n)_{[0,1]}$, $n\in\bN$, such that 
\begin{align}\label{eq:superexp def}
\lim_{n\to+\infty}\frac{1}{\log\dim\hil_n}\log\err_0(\rho_n|T_n)=+\infty=
\lim_{n\to+\infty}\frac{1}{\log\dim\hil_n}\log\err_1(\sigma_n|T_n),
\end{align}
which is equivalent to the symmetric error $\errm{1}{}{\rho_n}{\sigma_n}:=\min\left\{\err_0(\rho|T)+\err_1(\sigma|T)\,|\,T\in\B(\hil)_{[0,1]}\right\}=(1-\norm{\rho_n-\sigma_n}_1/2)/2$ decaying with a 
super-exponential speed, i.e., 
\begin{align*}
\lim_{n\to+\infty}\frac{1}{\log\dim\hil_n}\log\errm{1}{}{\rho_n}{\sigma_n}=+\infty.
\end{align*}

Below we show that for any quasi-free states $\omega_Q$ and $\omega_R$ on 
$\car{\ell^2_d(\bZ)}$, their local restrictions
$(\what\omega_{Q_n})_{n\in\bN}$ and 
$(\what\omega_{R_n})_{n\in\bN}$ can be super-exponentially discriminated. Equivalently, the 
quasi-free states can be super-exponentially discriminated by a sequence of tests
$T_n$ performed on length $n$ portions of the whole chain for every $n\in\bN$. More formally, there exists a sequence of tests $T_n\in\B\bz\ell^2(\n)\otimes\bC^d\jz_{[0,1]}$, $n\in\bN$, such that 
\begin{align}\label{eq:qf superexp def}
\lim_{n\to+\infty}\frac{1}{n}\log\Tr\what\omega_{Q_n}(I-T_n)=+\infty=
\lim_{n\to+\infty}\frac{1}{n}\log\Tr\what\omega_{R_n}T_n.
\end{align}
This is equivalent to \eqref{eq:superexp def} with $\rho_n:=\what\omega_{Q_n}$,
$\sigma_n:=\what\omega_{R_n}$ given on 
$\hil_n=\fock{\ell^2(\n)\otimes \bC^d}$, since in this case,
$\log\dim\hil_n=nd\log 2$.
 
Since quasi-free states are defined by single-particle symbol operators, it is not surprising that 
a sequence of tests with the above properties can be obtained from a sequence of operators on the single-particle Hilbert spaces $\ell^2(\n)\otimes\bC^d$.
Indeed, by Lemma 3.2 and Corollary 3.3 of \cite{superexp2023}, we have the following:

\begin{lemma}\label{lemma:tests from singlecopy projections}
For any projection $E_n=\sum_{k=1}^{r_n}\pr{e_{n,k}}$ on $\ell^2([n]^*)\otimes\bC^d$, 
where $(e_{n,k})_{k=1}^{r_n}$ is an orthonormal basis in $\ran E_n$, 
let $T_n\in\B\bz\fock{\ell^2([n]^*)\otimes\bC^d}\jz$
be the spectral projection of the number operator $N_{E_n}:=\sum_{k=1}^{r_n}a^*(e_{n,k})a(e_{n,k})$
corresponding to the eigenvalues $k=0,\ldots,\floor{\Tr E_n/2}$. Then, for any 
$G_n\in\B\bz\ell^2([n]^*)\otimes\bC^d\jz_{[0,1]}$,
\begin{align}
\omega_{G_n}\bz I-T_n\jz\le\bz\frac{8\Tr E_nG_n}{\Tr E_n}\jz^{\frac{\Tr E_n}{2}}
\ds\ds\text{and}\ds\ds\ds
\omega_{G_n}\bz T_n\jz\le\bz\frac{8\Tr E_n(I-G_n)}{\Tr E_n}\jz^{\frac{\Tr E_n}{2}}.
\label{eq:singleparticle upper bound}
\end{align}
\end{lemma}

Hence, our goal is to construct a sequence of projections $(E_n)_{n\in\bN}$ such that 
the upper bounds in \eqref{eq:singleparticle upper bound} decay super-exponentially fast for 
$G_n=Q_n$ in the first, and $G_n=R_n$ in the second inequality.
To construct such a sequence of projections, we will use an idea from 
\cite{superexp2023} based on discrete Fourier transform, and apply a suitable generalization of it
to the $d$ mode/site setting considered here. 

Recall that the discrete Fourier transform $F_n$ on $\ell^2(\n)$ is defined by its action on the canonical basis of $\ell^2(\n)$ as
\begin{align*}
    F_n \egy_{\{k\}} := \frac{1}{\sqrt{n}} \sum_{j=0}^{n-1} e^{i \frac{2\pi}{n} j k} \egy_{\{j\}}, \ds\ds\ds k \in \n.
\end{align*}
For the matrix-valued symbol case, we define $\F_{n} := F_n \otimes I_d$ acting on $\ran \P_n=
\ell^2(\n)\otimes\bC^d$. 

For convenience, we will identify the one-dimensional torus $\bT$ with $[-\pi,\pi)$ below instead of $[0,2\pi)$
as before, and consider the symbol functions of quasi-free states by first periodically extending them and then restricting to $[-\pi,\pi)$. The following is an extension of \cite[Lemma 3.5]{superexp2023} 
from the case of a scalar to the case of a matrix-valued symbol function. To state it, we will need the definition of the $n$-th Fej\'er kernel $\Phi_n$, given by 
\begin{align*}
\Phi_n(y):=\frac{1}{ n} \frac{\sin^2(ny/2)}{\sin^2(y/2)},\ds\ds\ds y\in[-\pi,\pi).
\end{align*}

\begin{lemma}\label{lemma:cesaro-diagonal}
Let $\mtxsymfn{g} \in L^\infty_{d\times d}(\bT)$ and $G = \ft^{-1}M_{\mtxsymfn{g}}\ft \in \B(\ell^2_d(\bZ))$ be the corresponding translation-invariant operator. Let $G_n = \P_n G \P_n$. The $d \times d$ diagonal blocks of $\F_{n} G_n \F_{n}^*$ in the canonical basis of $\ell^2(\n)$ are given by 
\begin{align*}
\left( \F_{n} G_n \F_{n}^* \right)_{k,k}
&:=
(\bra{\egy_{\{k\}}}\otimes I_d)\left( \F_{n} G_n \F_{n}^* \right)(\ket{\egy_{\{k\}}}\otimes I_d)\\
&= 
(\hat{S}_n \mtxsymfn{g})\left(\frac{2\pi k}{n}\right)
:= 
\frac{1}{2\pi}\int_{-\pi}^{\pi} \Phi_n(y) \mtxsymfn{g}\left(\frac{2\pi k}{n} - y\right) \dd y,
\end{align*}
for any $k \in \n$.
\end{lemma}
\begin{proof}
By definition, $G_n = \sum_{a,b=0}^{d-1} P_n G_{a,b} P_n \otimes \diad{a}{b}$, whence
for any $k \in \n$,
\begin{align*}
\left( \F_{n} G_n \F_{n}^* \right)_{k,k}
&=
\sum_{a,b=0}^{d-1} \inner{\egy_{\{k\}}}{F_nP_n G_{a,b} P_nF_n\egy_{\{k\}}} \otimes \diad{a}{b}\\
&=
\sum_{a,b=0}^{d-1}\left[\frac{1}{2\pi}\int_{-\pi}^{\pi} \Phi_n(y) 
\mtxsymfn{g}\left(\frac{2\pi k}{n} - y\right)_{a,b} \dd y\right]\otimes\diad{a}{b}\\
&=
\frac{1}{2\pi}\int_{-\pi}^{\pi} \Phi_n(y) \mtxsymfn{g}\left(\frac{2\pi k}{n} - y\right) \dd y,
\end{align*}
where the second equality follows from the scalar case proved in \cite[Lemma 3.5]{superexp2023}, and the 
last equality is by definition.
\end{proof}

For Lemma \ref{lemma:En construction} below, we will need some well-known properties of the Fej\'er kernel, which we summarize below for readers' convenience.
\begin{lemma}
For any $n\in\bN$, $n\ge 3$,  
\begin{align}
&\Phi_n(y)\le\min\{n,\pi^2/(ny^2)\},\ds\ds y\in[-\pi,\pi)\setminus\{0\},\label{eq:Fejer bound1}\\
&\int_{-\pi}^{\pi}\Phi_n(y)|y|\,dy\le 3\pi^2\frac{\log n}{n}.
\label{eq:Fejer bound2}
\end{align}
\end{lemma}
\begin{proof}
By symmetry, it is sufficient to consider $\Phi_n$ on $[0,\pi)$. 
For any $x\in[0,\pi)$ and $n\in\bN$, $|\sin(nx)|\le n|\sin x|$, as one can easily verify by induction on $n$, 
giving $\Phi_n(y)\le n$. On the other hand, for any $y\in[0,\pi)$, $\sin(y/2)\ge(y/2)(2/\pi)$, whence
$\Phi_n(y)\le (1/\sin(y/2))^2/n\le \pi^2/(ny^2)$. 
This proves \eqref{eq:Fejer bound1}, from which \eqref{eq:Fejer bound2} follows as
\begin{align*}
\int_{0}^{\pi}\Phi_n(y)y\,dy
&=
\int_{0}^{\pi/n}\underbrace{\Phi_n(y)}_{\le n}y\,dy
+
\int_{\pi/n}^{\pi}\underbrace{\Phi_n(y)}_{\le \pi^2/(ny^2)}y\,dy
\le
n\underbrace{\int_{0}^{\pi/n}y\,dy}_{=\pi^2/(2n^2)}
+
\frac{\pi^2}{n}\underbrace{\int_{\pi/n}^{\pi}\frac{1}{y}}_{=\log n}\,dy\\
&=\frac{\pi^2}{2n}\left[1+2\log n\right]\le \frac{3\pi^2}{2}\frac{\log n}{n}.
\end{align*}
\end{proof}

For a self-adjoint operator $A$, we will use the shorthand notation
\begin{align*}
\{A>0\}:=\sum_{a>0}P^A_a
\end{align*}
for the projection onto to the support of the positive part of $A$.

\begin{lemma}\label{lemma:En construction}
Let $\omega_Q$ and $\omega_R$ be translation-invariant quasi-free states defined by symbols $\mtxsymfn{q}, \mtxsymfn{r} \in L^\infty_{d\times d}(\bT)$.
Assume that there exists an interval $[\mu, \nu] \subset [0, 2\pi)$ of positive length such that
$\mtxsymfn{q}, \mtxsymfn{r}$ are Lipshitz-continuous on $[\mu, \nu]$, and one of the following holds:
\begin{enumerate}
\item\label{item:superexp1}
 For every $x\in[\mu, \nu]$, $\mtxsymfn{q}(x)$ and $\mtxsymfn{r}(x)$ are orthogonal, i.e., 
$\mtxsymfn{q}(x)\mtxsymfn{r}(x)=0$, and $\mtxsymfn{r}(x)$ is a non-zero projection.

\item\label{item:superexp2} 
For every $x\in[\mu, \nu]$, $I_d-\mtxsymfn{q}(x)$ and $I_d-\mtxsymfn{r}(x)$ are orthogonal, i.e., 
$(I_d-\mtxsymfn{q}(x))(I_d-\mtxsymfn{r}(x))=0$, and $I_d-\mtxsymfn{q}(x)$ is a non-zero projection.
\end{enumerate}
Then, for any $\delta\in(0,(\nu - \mu)/2)$, the operators 
\begin{align*}
E_{n,\delta} 
&:=
\bz F_n\otimes I_d\jz^*\Bigg(\sum_{k \in K_{n,\delta}} \diad{\egy_{\{k\}}}{ \egy_{\{k\}}} \otimes 
\{\mtxsymfn{r}(\theta_k)-\mtxsymfn{q}(\theta_k)>0\}\Bigg)\bz F_n\otimes I_d\jz,
\end{align*}
where 
\begin{align*}
    K_{n,\delta} := \left\{ k \in \n : \theta_k := \frac{2\pi k}{n} \in [\mu+\delta, \nu-\delta] \right\},
\end{align*}
are projections on $\ell^2(\n)\otimes\bC^d$ for every $n\in\bN$, and there exist constants
$c_0,c_1,c_2\in(0,+\infty)$ independent of $n$ such that 
\begin{align*}
\Tr E_{n,\delta}\ge c_0n,\ds\ds\ds
\Tr E_{n,\delta}Q_n\le c_1(\Tr E_{n,\delta})\frac{\log n}{n},\ds\ds\ds
\Tr E_{n,\delta}(I_{nd}-R_n)\le c_2(\Tr E_{n,\delta})\frac{\log n}{n},
\end{align*} 
for every large enough $n\in\bN$.
\end{lemma}
\begin{proof}
Let us fix a $\delta \in(0,(\nu - \mu)/2)$. 
Note that for any $k\in K_{n,\delta}$,
\begin{align*}
\{\mtxsymfn{r}(\theta_k)-\mtxsymfn{q}(\theta_k)>0\}
&=
\{(I_d-\mtxsymfn{q}(\theta_k))-(I_d-\mtxsymfn{r}(\theta_k))>0\}
=\begin{cases}
\mtxsymfn{r}(\theta_k),&\text{under assumption \ref{item:superexp1}},\\
I_d-\mtxsymfn{q}(\theta_k),&\text{under assumption \ref{item:superexp2}},
\end{cases}
\end{align*}
whence
\begin{align*}
E_{n,\delta} 
&=
\begin{cases}
\sum_{k \in K_{n,\delta}} \diad{F_n^* \egy_{\{k\}}}{F_n^* \egy_{\{k\}}} \otimes \mtxsymfn{r}(\theta_k),&
\text{under assumption \ref{item:superexp1}},\\
\sum_{k \in K_{n,\delta}} \diad{F_n^* \egy_{\{k\}}}{F_n^* \egy_{\{k\}}} \otimes \bz I_d-\mtxsymfn{q}(\theta_k)\jz,&
\text{under assumption \ref{item:superexp2}}.
\end{cases}
\end{align*}
In particular, $E_{n,\delta}$ is a projection under either assumption.
We prove the assertion under assumption \ref{item:superexp1}, since the proof under 
assumption \ref{item:superexp2} goes by an exactly analogous argument.

Assume therefore \ref{item:superexp1}.
Then $(\mtxsymfn{r}(x))_{x\in[\mu,\nu]}$ is a continuous family of non-zero projectors,
whence $p := \Tr(\mtxsymfn{r}(x))$ is constant on $[\mu,\nu]$ and is at least $1$.
Hence, 
\begin{align} \label{eqn:trace-En}
\Tr(E_{n,\delta}) = \sum_{k \in K_{n,\delta}} \Tr_{\bC^d}(\mtxsymfn{r}(\theta_k)) = p |K_{n,\delta}| 
\ge 
p \left\lfloor \frac{\nu - \mu - 2\delta}{2\pi} n \right\rfloor
\ge 
\underbrace{p \frac{\nu - \mu - 2\delta}{4\pi}}_{=:c_0} n, 
\end{align}
where the last inequality holds for every large enough $n$.
Moreover,
\begin{align}
\Tr(E_{n,\delta} Q_n) 
&= 
\sum_{k \in K_{n,\delta}}\sum_{a,b\in[d]^*}
\underbrace{\bz\Tr (Q_{a,b})_n\diad{F_n^* \egy_{\{k\}}}{F_n^* \egy_{\{k\}}}\jz}_{
=
(\hat{S}_n \mtxsymfn{q})\left(\theta_k\right)_{a,b} }
\underbrace{\Tr\diad{\egy_{\{a\}}}{\egy_{\{b\}}}\mtxsymfn{r}(\theta_k)}_{=\mtxsymfn{r}(\theta_k)_{b,a}}\nn\\
&=    
\sum_{k \in K_{n,\delta}} \Tr_{\bC^d}\left[ (\hat{S}_n \mtxsymfn{q})(\theta_k) \mtxsymfn{r}(\theta_k)\right] 
\nn\\
&= 
\sum_{k \in K_{n,\delta}} \frac{1}{2\pi} \int_{-\pi}^{\pi} \Phi_n(y) 
\Tr_{\bC^d}\left[ \mtxsymfn{q}(\theta_k - y) \mtxsymfn{r}(\theta_k) \right] \dd y\nn\\
&= \sum_{k \in K_{n,\delta}} \frac{1}{2\pi} \int_{-\pi}^{\pi} \Phi_n(y) 
\Tr_{\bC^d}\left[\left( \mtxsymfn{q}(\theta_k - y) - \mtxsymfn{q}(\theta_k))\mtxsymfn{r}(\theta_k)
\right) \right] \dd y,\label{eq:superexp proof1}
\end{align}
where the first and the third equalities follow from Lemma \ref{lemma:cesaro-diagonal},
the second equality is trivial, and in the last equality we used the assumption that 
$\mtxsymfn{q}(x)$ and $\mtxsymfn{r}(x)$ are orthogonal for every $x\in[\mu,\nu]$.

We bound the integrals by splitting them each into two parts. 
First, note that for any $\ep\in(0,\pi)$ and any 
$y\in[-\pi,\pi)\setminus[-\ep,\ep]$,
\begin{align}
\Phi_n(y) \le \frac{\gamma_{\ep}}{n},\ds\ds\text{where}\ds\ds \gamma_{\ep}:=\frac{1}{\sin^2(\ep/2)}.
\end{align}
Hence, for any $k\in K_{n,\delta}$,
\begin{align}
\int_{[-\pi,\pi)\setminus[-\delta,\delta]}\underbrace{\Phi_n(y)}_{\le\gamma_{\delta}/n} 
\underbrace{
\Tr_{\bC^d}\left[\mtxsymfn{q}(\theta_k - y)\mtxsymfn{r}(\theta_k) \right]}_{\le\Tr\mtxsymfn{r}(\theta_k)=p} \dd y
\le
\frac{2p(\pi-\ep)\gamma_{\delta}}{n}
\le
\frac{2p\pi\gamma_{\delta}}{n}
\,.\label{eq:superexp proof2}
\end{align}
Next, let $L_{\mtxsymfn{q}}:=\max\{\norm{\mtxsymfn{q}(x)-\mtxsymfn{q}(y)}/|x-y|,\,x,y\in[\mu,\nu],\,x\ne y\}$
be the Lipschitz constant of $\mtxsymfn{q}$ on $[\mu,\nu]$. Then, 
for any $k\in K_{n,\delta}$,
\begin{align}
\int_{-\delta}^{\delta} \Phi_n(y) 
\underbrace{\Tr_{\bC^d}\left[\left( \mtxsymfn{q}(\theta_k - y) - \mtxsymfn{q}(\theta_k)\right)\mtxsymfn{r}(\theta_k)\right]}_{\le L_{\mtxsymfn{q}}|y|\Tr\mtxsymfn{r}(\theta_k)} \dd y
\le
L_{\mtxsymfn{q}}\Tr\mtxsymfn{r}(\theta_k)\underbrace{\int_{-\delta}^{\delta} \Phi_n(y)|y|\dd y}_{\le 3\pi^2\frac{\log n}{n}}
\le
3\pi^2pL_{\mtxsymfn{q}}\frac{\log n}{n},
\label{eq:superexp proof3}
\end{align}
where the last inequality is due to \eqref{eq:Fejer bound2} and it holds for every $n\ge 3$.
Putting together \eqref{eq:superexp proof1}--\eqref{eq:superexp proof3} then yields
\begin{align*}
\Tr(E_{n,\delta} Q_n) 
&\le
\underbrace{p|K_{n,\delta}|}_{=\Tr E_{n,\delta}}\frac{1}{p}\left[\frac{2p\pi\gamma_{\delta}}{n}+3\pi^2pL_{\mtxsymfn{q}}\frac{\log n}{n}\right]
\le c_1(\Tr E_{n,\delta})\frac{\log n}{n}
\end{align*}
for $c_1:=2\max\{2\pi\gamma_{\delta},3\pi^2L_{\mtxsymfn{q}}\}$ and every $n\ge 3$.

Similarly, we have 
\begin{align}
    \Tr(E_{n,\delta} (I_{nd} - R_n)) &= \sum_{k \in K_{n,\delta}} \Tr_{\bC^d}\left[ \mtxsymfn{r}(\theta_k) \left( I_d - (\hat{S}_n \mtxsymfn{r})(\theta_k) \right) \right] \nn\\
    &= \sum_{k \in K_{n,\delta}} \frac{1}{2\pi} \int_{-\pi}^{\pi} \Phi_n(y) \Tr_{\bC^d}\left[ \mtxsymfn{r}(\theta_k) \left( I_d - \mtxsymfn{r}(\theta_k - y) \right) \right] \dd y\nn\\
& = \sum_{k \in K_{n,\delta}} \frac{1}{2\pi} \int_{-\pi}^{\pi} \Phi_n(y) \Tr_{\bC^d}\left[ \mtxsymfn{r}(\theta_k) \left( \mtxsymfn{r}(\theta_k) - \mtxsymfn{r}(\theta_k - y) \right) \right] \dd y,
\label{eq:superexp proof4}
\end{align}
where the first two equalities follow by Lemma \ref{lemma:cesaro-diagonal}, and in the last equality we used
 the assumption that $\mtxsymfn{r}(x)$ is a projection at every $x\in[\mu,\nu]$.
The integrals can also be bounded similarly to the above, as
\begin{align}
\int_{[-\pi,\pi)\setminus[-\delta,\delta]}\underbrace{\Phi_n(y)}_{\le\gamma_{\delta}/n} 
\underbrace{\Tr_{\bC^d}
\left[\mtxsymfn{r}(\theta_k)(I_d-\mtxsymfn{r}(\theta_k - y))\right]}_{\le\Tr\mtxsymfn{r}(\theta_k)=p} \dd y
\le
\frac{2p(\pi-\ep)\gamma_{\delta}}{n}
\le
\frac{2p\pi\gamma_{\delta}}{n}
\,,\label{eq:superexp proof5}
\end{align}
and
\begin{align}
\int_{-\delta}^{\delta} \Phi_n(y) 
\underbrace{\Tr_{\bC^d}\left[\mtxsymfn{r}(\theta_k)\bz\mtxsymfn{r}(\theta_k - y) - \mtxsymfn{r}(\theta_k)\jz
\right]}_{\le L_{\mtxsymfn{r}}|y|\Tr\mtxsymfn{r}(\theta_k)} \dd y
\le
L_{\mtxsymfn{r}}\Tr\mtxsymfn{r}(\theta_k)\underbrace{\int_{-\delta}^{\delta} \Phi_n(y)|y|\dd y}_{\le 3\pi^2\frac{\log n}{n}}
\le
3\pi^2pL_{\mtxsymfn{r}}\frac{\log n}{n},
\label{eq:superexp proof6}
\end{align}
where $L_{\mtxsymfn{r}}:=\max\{\norm{\mtxsymfn{r}(x)-\mtxsymfn{r}(y)}/|x-y|,\,x,y\in[\mu,\nu],\,x\ne y\}$
is the Lipschitz constant of $\mtxsymfn{r}$ on $[\mu,\nu]$, and
the last inequality is due to \eqref{eq:Fejer bound2}  and it holds for every $n\ge 3$.
Putting together \eqref{eq:superexp proof4}--\eqref{eq:superexp proof6} then yields
\begin{align*}
\Tr(E_n(I_{nd}- R_n)) 
&\le
\underbrace{p|K_{n,\delta}|}_{=\Tr E_{n,\delta}}\frac{1}{p}\left[\frac{2p\pi\gamma_{\delta}}{n}+3\pi^2pL_{\mtxsymfn{r}}\frac{\log n}{n}\right]
\le c_2(\Tr E_{n,\delta})\frac{\log n}{n}
\end{align*}
for $c_2:=2\max\{2\pi\gamma_{\delta},3\pi^2L_{\mtxsymfn{r}}\}$ and every $n\ge 3$.
\end{proof}

\begin{theorem}
In the setting of Lemma \ref{lemma:En construction}, 
consider the projections $E_{n,\delta}$, $n\in\bN$, for some $\delta\in(0,(\nu-\mu)/2)$, and let 
$T_{n,\delta}\in\B\bz\fock{\ell^2(\n)\otimes\bC^d}\jz$, $n\in\bN$, be the tests constructed from 
$E_{n,\delta}$ as in Lemma \ref{lemma:tests from singlecopy projections}. Then there exists a constant $c\in(0,+\infty)$ such that 
\begin{align}\label{eq:superexp bound}
\Tr\what\omega_{Q_n}(I_{nd}-T_{n,\delta})\le e^{-cn\log n},\ds\ds\ds
\Tr\what\omega_{R_n}T_{n,\delta}\le e^{-cn\log n},
\end{align} 
for every large enough $n$.
In particular, $\omega_Q$ and $\omega_R$ can be super-exponentially discriminated in the sense of 
\eqref{eq:qf superexp def}.
\end{theorem}
\begin{proof}
For every large enough $n$, 
\begin{align*}
\omega_{Q_n}\bz I-T_{n,\delta}\jz
&\le
\bz\frac{8\Tr E_{n,\delta}Q_n}{\Tr E_{n,\delta}}\jz^{\frac{\Tr E_{n,\delta}}{2}}
\le
\bz 8c_1\frac{\log n}{n}\jz^{\frac{\Tr E_{n,\delta}}{2}}
\le
\bz 8c_1\frac{\log n}{n}\jz^{\frac{c_0}{2}n}\\
&=
e^{-\frac{c_0}{2}n\left[\log n-\log\log n-\log (8c_1)\right]}\le e^{-(c_0/4)n\log n},
\end{align*}
where the first inequality is due to 
Lemma \ref{lemma:tests from singlecopy projections}, the second inequality follows from 
Lemma \ref{lemma:En construction}, the third inequality holds for every $n\in\bN$ such that 
$(8c_1\log n)/n\le 1$, and the last inequality is true when $n$ is large enough so that 
$\log\log n+\log(8c_1)<(1/2)\log n$. By a completely analogous argument,
\begin{align*}
\omega_{R_n}\bz T_{n,\delta}\jz
&\le
\bz\frac{8\Tr E_{n,\delta}(I-R_n}{\Tr E_{n,\delta}}\jz^{\frac{\Tr E_{n,\delta}}{2}}
\le
\bz 8c_2\frac{\log n}{n}\jz^{\frac{\Tr E_{n,\delta}}{2}}
\le
\bz 8c_2\frac{\log n}{n}\jz^{\frac{c_0}{2}n}\\
&=
e^{-\frac{c_0}{2}n\left[\log n-\log\log n-\log (8c_2)\right]}\le e^{-(c_0/4)n\log n}
\end{align*}
for every large enough $n$.
Thus, \eqref{eq:superexp bound} holds with $c:=c_0/4$ for every large enough $n$.
The statement about super-exponential discrimination follows immdediately from 
\eqref{eq:superexp bound}.
\end{proof}

\section*{Acknowledgments}

This work was partially funded by the
National Research, Development and 
Innovation Office of Hungary (NKFIH) via the research grants K 146380 and EXCELLENCE 151342, and
by the Ministry of Culture and Innovation and the National Research, Development and Innovation Office within the Quantum Information National Laboratory of Hungary (Grant No. 2022-2.1.1-NL-2022-00004).
MM was partially supported by the Ministry of Education, Singapore, through grant T2EP20124-0005.
GMZ was partially supported by the QuantERA II project HQCC-101017733 (Grant No. 2019-2.1.7-ERA-NET-2022-00052).
The authors are grateful to Zolt\'an Zimbor\'as for discussions.

\bibliography{bibliography_blocktop}

\end{document}